\documentclass[journal,10pt]{IEEEtran}
\usepackage[T1]{fontenc}
\usepackage{amsmath,amssymb,mathtools}
\usepackage{cite}
\usepackage{url}
\usepackage{graphicx}
\usepackage{placeins}
\usepackage{flafter}
\usepackage{array,multirow}
\usepackage{microtype}
\usepackage[hidelinks]{hyperref}
\allowdisplaybreaks
\newtheorem{theorem}{Theorem}
\newtheorem{lemma}{Lemma}
\newtheorem{corollary}{Corollary}
\newtheorem{proposition}{Proposition}
\newtheorem{remark}{Remark}
\newtheorem{definition}{Definition}
\newtheorem{example}{Example}
\newcommand{\cN}{\mathcal{N}}
\newcommand{\cB}{\mathcal{B}}
\newcommand{\cG}{\mathcal{G}}
\newcommand{\cH}{\mathcal{H}}
\newcommand{\cS}{\mathcal{S}}
\newcommand{\cD}{\mathcal{D}}
\newcommand{\bd}{\mathbf{d}}
\newcommand{\abs}[1]{\left|#1\right|}
\newcommand{\given}{\,|\,}
\hypersetup{pdftitle={Secure D2D Coded Caching under Radius-Limited Connectivity via Lifting},pdfauthor={Han Fang, Nan Liu, Wei Kang}}
\begin{document}
\title{Secure D2D Coded Caching under Radius-Limited Connectivity via Lifting}
\author{Han~Fang, Nan~Liu, and Wei~Kang%
\thanks{Han Fang is with the School of Cyber Science and Engineering,
Southeast University, Nanjing, China (e-mail: h\_fang@seu.edu.cn).}%
\thanks{Nan Liu is with the National Mobile Communications Research Laboratory
and the School of Cyber Science and Engineering, Southeast University,
Nanjing, China (e-mail: nanliu@seu.edu.cn).}%
\thanks{Wei Kang is with the School of Information Science and Engineering
and the School of Cyber Science and Engineering, Southeast University,
Nanjing, China (e-mail: wkang@seu.edu.cn).}}
\maketitle
\begin{abstract}
Secure device-to-device (D2D) caching requires users to generate delivery
signals from their own caches and to recover their requested files without
learning the other files. We construct schemes that meet both requirements
when communication is limited to nearby users. The network consists of
\(K\) users on a cycle, a library of \(N\) independent files, and a reception
radius \(r\). Delivery uses one round of noiseless orthogonal local broadcasts,
with no server transmission. Each user's cache has capacity \(M\) file units.
An outer secret-sharing code assigns file shares to senders, and a secure
local caching scheme at each sender delivers the requested shares. The cache
budget includes the sender's library representation, random state, and
delivery keys. We prove that each receiver's cache and all audible broadcasts
jointly reveal no information about nonrequested files, despite the
dependence between the outer shares. The converse counts transmissions
crossing a contiguous user block and uses
secrecy to remove one user's observation from the entropy bound. Let
\(h=2r\) and \(K\ge h+1\). The private coded caching (PCC) instantiation
attains the minimum peak sum load \(K/h\) for every \(M\ge Nh\) when
\(N>1\). For \(N\ge K\), its memory-sharing envelope is within a factor
\(15\) of the optimum throughout \(M\in[N/h+2,Nh]\).
We also apply lifting without secrecy constraints, achieving the optimal
load on \(M\in[hN/(h+1),N]\) and a factor-\(9/2\) guarantee on
\(M\in[N/(h+1),N]\) when \(N\ge K\).
\end{abstract}


\section{Introduction}

In secure D2D caching, a user must recover its requested file from nearby
devices without learning the other files from its cache and received
signals. Every device must also be able to generate its transmission from
the content and keys placed in its cache before the requests arrive. We
study these requirements together on a cycle, where each transmission reaches
the \(h=2r\) neighbors within radius \(r\). Our construction places a secure
shared-link caching instance at each sender. We include the storage needed
to run these instances and prove secrecy for each receiver's cache and all
broadcasts it hears.

D2D communication supports short-range content delivery through nearby
devices \cite{DopplerRinneWijtingRibeiroHugl2009,AsadiWangMancuso2014,
TehraniUysalYanikomeroglu2014}. Placing popular content at these devices
reduces the traffic carried by cellular infrastructure
\cite{GolrezaeiMolischDimakisCaire2013,GolrezaeiDimakisMolisch2014}.
Coded caching creates multicast transmissions that serve several requests at
once. The Maddah-Ali--Niesen (MAN) construction achieves this gain by
coordinating cache placement and delivery over a shared link
\cite{MaddahAliNiesen2014}. Its server holds the entire library. In the D2D
model studied here, the server is silent during delivery, and each user has
a finite cache and can reach only nearby users. We minimize the worst-case
sum of all user transmission entropies, normalized by file size.

The construction has two layers. An outer secret-sharing code assigns a
share of every file to each sender. A single sender's share library reveals
no file content, and a receiver can reconstruct its requested file from
its home share and the \(h\) requested shares delivered by its neighbors.
Each sender uses a secure local caching scheme to deliver those shares.
The shares at different senders are correlated because they encode the same
files. We handle this dependence by showing that, conditioned on the
requested share, a local receiver's view has a distribution independent of
the other local files. Independent random seeds at the senders then give
a joint factorization that proves secrecy for the entire received view.

For a contiguous block of users, the converse counts the external
transmitters whose broadcasts reach the block. Conditioned on an anchor
user's requested file, its cache and audible broadcasts reveal no information
about the remaining files. We therefore remove this observation from the
entropy bound and bound the load using the remaining transmissions.

\subsection{Contributions and Main Results}

Let \(N\) be the number of files, \(K\) the number of users, and
\(h=2r\) the number of neighbors heard by each user, with \(K\ge h+1\).
Write \(M\) for each user's cache capacity and \(R\) for the peak sum load,
both normalized by file size. The main results are as follows.

\emph{Secure lifting and joint secrecy.}
A local scheme has receiver memory \(M_{\rm rx}\), sender memory
\(M_{\rm snd}\), and broadcast load \(R_{\rm loc}\), each normalized by
the local file size. It gives the D2D point
\[
M=M_{\rm rx}+\frac{M_{\rm snd}}{h},
\qquad R=\frac{K}{h}R_{\rm loc}.
\]
The sender memory includes the recoverable local library representation,
random state, and delivery keys. Theorem~\ref{thm:tifs-secure-lift} proves
this map for local schemes satisfying Definition~\ref{def:tifs-composable-kernel}.
Lemma~\ref{lem:tifs-multi-sender} proves joint file secrecy across the
correlated shares and all audible senders.

\emph{A radius-dependent lower bound.}
Theorem~\ref{thm:tifs-secure-converse} bounds the load using contiguous user
blocks. For a block of \(s\in\{1,\ldots,K-1\}\) users, let \(\delta_s\)
be the number of external senders audible to the block. Let \(\eta_s\)
be the largest number of these senders heard by one user in the block.
We choose a user attaining this maximum as the anchor. The two counts are
\(\delta_s=\min\{h,K-s\}\) and
\(\eta_s=\min\{K-s,\max\{h/2,h+1-s\}\}\).
The converse subtracts the external broadcasts already included in the
anchor's observation.

\emph{Exact high-memory load and a factor-\(15\) guarantee.}
Let \(R_{\rm sec}^{\star}(M,r)\) denote the optimal peak sum load allowed
by reliability and internal file secrecy at memory \(M\) and radius \(r\).
The private coded caching (PCC) scheme of
\cite{RavindrakumarPandaKaramchandaniPrabhakaran2018} gives the secure
family in Theorem~\ref{thm:tifs-pcc-lift}.
The PCC-lift family achieves
\[
R_{\rm sec}^{\star}(M,r)=K/h,
\qquad M\ge Nh,
\]
for \(N>1\). Write \(R_{{\rm PCC}\to r}^{\rm up}(M)\) for the lowest load
obtained by memory sharing among the PCC-lift points, represented by their
lower convex envelope. For \(N\ge K\), this achievable load satisfies
\[
R_{{\rm PCC}\to r}^{\rm up}(M)
\le 15R_{\rm sec}^{\star}(M,r),
\qquad M\in[N/h+2,Nh].
\]
The factor is independent of the library size, number of users, and radius.

A maximum-distance-separable (MDS) outer code also lifts shared-link schemes
to the non-secure model. For the MAN scheme, the lifted load is optimal on
\(M\in[hN/(h+1),N]\). When \(N\ge K\), memory sharing is within a factor
\(9/2\) of the optimum on \(M\in[N/(h+1),N]\).
Section~\ref{sec:nonsecure-counterpart} gives the construction and proofs.

\subsection{Related Work}

Shared-link caching supplies the local schemes used in our construction.
Decentralized placement achieves an order-optimal tradeoff
\cite{MaddahAliNiesen2015Decentralized}, and the optimal rate under uncoded
prefetching is known \cite{YuMaddahAliAvestimehr2018}.
Coded placement improves achievable rates in several memory regimes
\cite{ChenFanLetaief2016,AmiriGunduz2017,GomezVilardebo2018,TianChen2018,
ZhangTian2018}. Placement delivery arrays (PDAs), hypergraphs, and linear block
codes provide constructions that control the number of packets per file
(subpacketization)
\cite{YanChengTangChen2017PDA,ShangguanZhangGe2018,TangRamamoorthy2018,
AravindSarvepalliThangaraj2022}. Our lifting theorem specifies the local
secrecy and sender-storage conditions needed to use a shared-link scheme
at a physical D2D sender.

Ji, Caire, and Molisch studied coded delivery in fully connected D2D
networks \cite{JiCaireMolisch2016} and the throughput and system performance
of D2D caching \cite{JiCaireMolisch2016JSAC}.
Their framework also considers finite transmission distance and spatial reuse
\cite{JiCaireMolisch2016}.
Subsequent results characterize uncoded placement with one-shot delivery
\cite{YaparWanSchaeferCaire2019}, handle unequal cache capacities
\cite{IbrahimZewailYener2020}, and reduce packetization
\cite{WangChengYanTang2019D2DPDA,WoolseyChenJi2020}.
Here the fixed reception radius determines the shares available for decoding,
the joint observation subject to internal file secrecy, and the
transmissions counted by the block-boundary converse.

Multi-access caching studies users that read from several helper caches
\cite{HachemKaramchandaniDiggavi2017,ReddyKaramchandani2020,
ChengWanLiangZhangCaire2021,ChengWanEliaCaire2025}.
D2D-assisted schemes have been developed for cyclic and combinatorial
access topologies and for linear subpacketization
\cite{WuChengChenWuChen2024,UmmerRajan2025,UmmerRajan2026,
WangChengWuWang2026}. Remark~\ref{rem:tifs-model-boundary} distinguishes
helper-cache access from D2D reception.

Secure caching has considered both external listeners
\cite{SenguptaTandonClancy2015,AwanSezgin2015,AwanMathar2018}
and legitimate users who should learn only their requested files.
For the latter requirement, shared-link constructions are available in
\cite{RavindrakumarPandaKaramchandaniPrabhakaran2018}, and fully connected
D2D constructions in \cite{ZewailYener2016Asilomar,ZewailYener2020}.
PDA-based and shared-cache constructions provide further designs
\cite{MeelRajan2021PDA,MeelRajan2021Shared}.
Related security topics include demand privacy against colluding users
\cite{WanSunJiTuninettiCaire2022} and proactive D2D cache refresh under cache
attacks \cite{ZhangLinCheng2026}.
Our lifting theorem addresses internal file secrecy with finite user
caches and radius-limited broadcasts.

\begin{remark}[Cache access and D2D reception]
\label{rem:tifs-model-boundary}
A helper-cache access radius specifies which storage nodes a user can read.
A D2D reception radius specifies which user broadcasts it receives.
Here each user has one cache and uses its contents for transmission and
decoding.
\end{remark}

\subsection{Full-Connectivity Specialization}

When \(K=h+1\), every user hears all other users. Index the secure PCC-lift
points by an integer \(u\in\{1,\ldots,K-2\}\). Their coordinates reduce to
\[
(M,R)=\left(\frac{Nu}{K-u}+1+\frac1u,\frac Ku\right),
\]
together with
\[
(M,R)=\left(N(K-1),\frac{K}{K-1}\right).
\]
These are the centralized-placement secure D2D points of Zewail and Yener
\cite{ZewailYener2020}. The non-secure MAN lift reduces to the D2D points
of Ji, Caire, and Molisch \cite{JiCaireMolisch2016}. Theorems
\ref{thm:tifs-secure-lift} and~\ref{thm:tifs-ns-lift} give the
corresponding constructions for every \(K\ge h+1\).

\section{System Model}
\label{sec:system-model}

\subsection{Secure Radius-Limited D2D Model}
For a positive integer \(m\), write \([m]=\{1,\ldots,m\}\).
For any indexed family \((U_i)\) and index set \(\mathcal A\), write
\(U_{\mathcal A}=(U_i:i\in\mathcal A)\).
There are \(K\) users on a cycle and \(N\) files, where \(K,N\) are
positive integers. Fix an integer reception radius \(r\ge1\), put
\(h:=2r\), and assume
\begin{equation}
K\ge h+1=2r+1.
\label{eq:tifs-standing-K}
\end{equation}
The files \(W_1,\ldots,W_N\) are mutually independent random variables.
Let \(F>0\) be their common entropy in bits. Each file is uniform over
a finite alphabet of size \(2^F\). User indices are taken modulo \(K\).
User \(k\) transmits to and
receives from the radius-\(r\) neighborhood
\[
    \cN_r(k):=\{k-r,\ldots,k+r\},
\]
which contains \(h+1\) users: user \(k\) and its \(h\) cyclic neighbors.
A server initializes the caches before delivery, and each user generates
its broadcast from its own cache.

The server is trusted, and each user is an individual honest-but-curious
observer that follows the prescribed encoding maps.  One communication instance
consists of placement, a public demand vector, and one non-interactive delivery
round over noiseless orthogonal local broadcasts.
Let \(M\ge0\) be the cache budget per user, normalized by file size.
Placement uses a random seed \(Q\) independent of \(W_{[N]}\).
The placement encoder maps \((W_{[N]},Q)\) to the random caches
\(Z_1,\ldots,Z_K\), with \(H(Z_k)\le MF\).
For each fixed public demand vector \(\bd=(d_1,\ldots,d_K)\in[N]^K\),
\(d_k\) is the index of the file requested by user \(k\).
Sender \(k\) uses a deterministic local delivery encoder \(\psi_k\)
to generate the broadcast random variable
\[
    X_k^{\bd}=\psi_k(Z_k,\bd).
\]
Let \(R\ge0\) be a bound on the peak sum load, normalized by file size.
The load constraint requires, for every demand,
\begin{equation}
    \sum_{k=1}^K H(X_k^\bd)\le RF.
\label{eq:tifs-sum-load}
\end{equation}
For user \(u\in[K]\), the complete local view is
\(\bigl(Z_u,\{X_k^{\bd}:k\in\cN_r(u)\},\bd\bigr)\).  Reliability requires
\[
H\!\left(W_{d_u}\given Z_u,\{X_k^{\bd}:k\in\cN_r(u)\},\bd\right)=0
\]
for every \(u\) and every demand.
Perfect internal file secrecy requires
\begin{equation}
I\!\left(W_{[N]\setminus\{d_u\}};
Z_u,\{X_k^{\bd}:k\in\cN_r(u)\}\given \bd\right)=0.
\label{eq:tifs-file-sec}
\end{equation}
Reliability and file independence give the equivalent requested-file-conditioned
form
\begin{equation}
I\!\left(W_{[N]\setminus\{d_u\}};
Z_u,\{X_k^{\bd}:k\in\cN_r(u)\}\given W_{d_u},\bd\right)=0.
\label{eq:tifs-file-sec-conditioned}
\end{equation}
Indeed, for \(A=W_{[N]\setminus\{d_u\}}\), \(B=W_{d_u}\), and the complete
observation \(V\), file independence gives \(I(A;B)=0\) and recovery gives
\(H(B\mid V,\bd)=0\).  The chain rule then yields
\(I(A;V\mid B,\bd)=I(A;V\mid\bd)\).

Figure~\ref{fig:tifs-system-model} illustrates the cyclic topology,
cache placement, and local D2D communication.
\begin{figure}[t]
\centering
\includegraphics[width=\columnwidth]{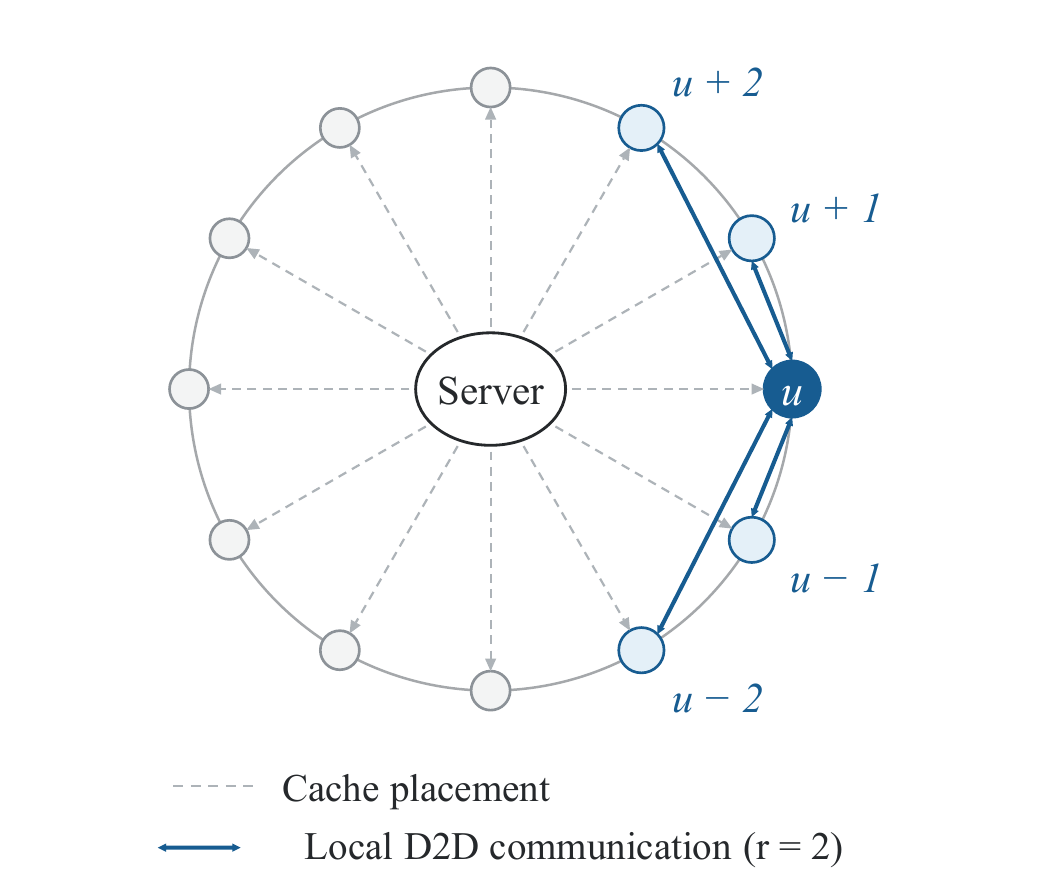}
\caption{Cyclic D2D topology, illustrated for \(r=2\). The central server
initializes the user caches before delivery. User \(u\) communicates with
its \(h=2r\) cyclic neighbors through local D2D broadcasts.}
\label{fig:tifs-system-model}
\end{figure}

For a fixed admissible file size \(F\), an \emph{exactly achievable} pair
\((M,R)\) is one for which placement and non-interactive delivery maps satisfy
the stated cache and peak-load bounds, zero-error reliability, and perfect file
secrecy for every demand.  The closure of these pairs over all admissible file
sizes \(F\) defines the secure memory--load region, with optimal load
\[
R_{\rm sec}^\star(M,r)
:=\inf\left\{R:\ (M,R)\text{ is in the secure region}\right\}.
\]
The region is convex by independent concatenation. Consider codes with
file sizes \(F_1,F_2\), and let \(a,b\) be positive integers.
Concatenating \(a\) independent copies of the first code and \(b\) of the
second preserves recovery and secrecy and gives the resource weight
\(aF_1/(aF_1+bF_2)\). For any target weight \(\lambda\in[0,1]\),
a sequence of choices of \(a,b\) makes this ratio converge to \(\lambda\).
Taking the closure gives memory sharing at every weight.
For a real number \(x\), define \([x]^+:=\max\{x,0\}\).

\begin{table*}[t]
\centering
\caption{Frequently used notation.}
\label{tab:tifs-notation}
\scriptsize
\setlength{\tabcolsep}{4pt}
\renewcommand{\arraystretch}{1.10}
\begin{tabular}{p{1.65cm}p{6.65cm}p{1.65cm}p{6.65cm}}
\hline
Symbol & Meaning & Symbol & Meaning \\
\hline
\(N,K\) & Number of files and cyclic users & \(r,h=2r\) & Reception radius and number of other users in a full radius window \\
\(W_n,F\) & File \(n\) and its entropy in bits & \(M,R\) & Normalized cache size and peak total sum load \\
\(Z_k,X_k^{\bd}\) & Cache and non-interactive broadcast of user \(k\) & \(\cN_r(k)\) & Users whose broadcasts are heard by user \(k\), including itself \\
\(\bd,\bd^{(k)}\) & Global demand and its projection to sender \(k\)'s local receivers & \(S_{n,k}\) & Outer one-private share of file \(n\) assigned to sender position \(k\) \\
\(C_u^{(k)},G_k\) & Receiver cache in local instance \(k\), and sender-side state & \(V_u^{(k)}\) & Local view \((C_u^{(k)},X_k)\) of receiver \(u\) in instance \(k\) \\
\(t,u=t+1\) & PCC parameter and re-indexed achievable-point parameter & \(q_s\) & \(\lfloor N/s\rfloor\), number of demand instances in the secure converse \\
\(\delta_s\) & Number of external boundary senders audible to an \(s\)-user block & \(\eta_s\) & Maximum boundary senders included in one anchor user's secure observation \\
\hline
\end{tabular}
\end{table*}
\begin{remark}[Global delivery secrecy of the constructed points]
The keyed PCC-lift family and the high-memory endpoint satisfy global delivery
secrecy,
\[
I(W_{[N]};X_1^\bd,\ldots,X_K^\bd\given\bd)=0.
\]
Independent keys mask the complete delivery vector at each keyed PCC-lift
point.  At the high-memory endpoint, the linear sharing construction uses
independent randomness for every sender--file pair.
Proposition~\ref{prop:tifs-endpoint-delivery-secrecy} proves secrecy of the joint
transcript across all senders.
\end{remark}

\subsection{Non-Secure Benchmark Model}
The non-secure radius-\(r\) D2D benchmark optimizes zero-error recovery under
the same cyclic topology, cache budget, local non-interactive delivery maps,
and peak total sum load.  Its achievable region is the closure of the
finite-file-size memory--load pairs satisfying these placement, delivery, and
reliability conditions, and
\(R^\star(M,r)\) denotes the infimum peak load at memory \(M\).

The converse proofs compare user blocks and different file requests.
The following symmetrization makes the corresponding entropy calculations
invariant under cyclic user shifts and file relabelings.

\begin{lemma}[Cyclic-user and file-index symmetrization]
\label{lem:tifs-symmetrization}
Under this single-shot non-interactive model, every secure or non-secure
scheme with memory--load pair \((M,R)\) can be converted, without changing
\((M,R)\), into a scheme invariant under all cyclic user shifts and all file
permutations.  Zero-error recovery is preserved, and in the secure case so is
perfect file secrecy.  Consequently, the optima
\(R_{\rm sec}^\star(M,r)\) and \(R^\star(M,r)\) are unchanged if the
optimization is restricted to schemes with both symmetries.
\end{lemma}

\begin{IEEEproof}
Let the original code have file size \(F_0\), let \(\mathfrak S_N\) be the
file-permutation group, and set \(L=K|\mathfrak S_N|\).  Form new files of
size \(F'=LF_0\), each containing independent subfiles
\(W_n^{(a,\pi)}\) of size \(F_0\), indexed by
\((a,\pi)\in\{0,\ldots,K-1\}\times\mathfrak S_N\).
Here \(a\) is a cyclic shift and \(\pi\) is a file permutation.
Shift labels are added modulo \(K\).
Run an independent copy of
the original code in each component: physical user \(j+a\) plays baseline
role \(j\), and baseline file \(m\) is \(W_{\pi(m)}^{(a,\pi)}\).  Under
physical demand \(\bd\), baseline role \(j\) therefore receives demand
\[
    \widetilde d_j^{(a,\pi)}
    =\pi^{-1}(d_{j+a}).
\]
User \(u\) concatenates the caches and broadcasts assigned to role \(u-a\)
in each component.  Entropy subadditivity gives
\[
H(Z'_u)\le LMF_0=MF',\quad
\sum_u H(X_u'^\bd)\le LRF_0=RF'.
\]
Cyclic relabeling preserves each radius window, and every component broadcast
is a function of its sender's component cache and demand.

In component \((a,\pi)\), physical user \(u\) decodes baseline file
\(\pi^{-1}(d_u)\), namely the subfile \(W_{d_u}^{(a,\pi)}\).  Concatenating
over all components reconstructs \(W_{d_u}\).  In the secure case, the
original secrecy constraint applied to baseline role \(u-a\) makes that
component observation independent of all subfiles
\(W_n^{(a,\pi)}\) with \(n\ne d_u\).  Components use independent file
subblocks and independent placement randomness, so their conditional laws
factor, and the concatenated observation has zero leakage.  A cyclic user
shift by \(b\) permutes components by \((a,\pi)\mapsto(a+b,\pi)\), while a
file relabeling \(\sigma\) permutes them by
\((a,\pi)\mapsto(a,\sigma\circ\pi)\).  These relabelings, together with the
corresponding demand relabeling, give both symmetries.
\end{IEEEproof}

\section{Main Secure Results}
\label{sec:tifs-main-results}

Each physical sender runs one local caching instance. Its cache stores the
state used to serve that instance and the receiver caches from neighboring
instances.

\begin{definition}[Composable local kernel with a recoverable sender library]
\label{def:tifs-composable-kernel}
Consider \(h\) receivers and independent local files \(A_1,\ldots,A_N\), each
uniformly distributed over a finite alphabet of size \(2^B\), where \(B>0\).
A local kernel consists of receiver caches \(C_1,\ldots,C_h\), a sender state
\(G\), and a broadcast \(X^{\mathbf e}\) for each public demand
\(\mathbf e=(e_1,\ldots,e_h)\in[N]^h\).
Let \(\Phi\) and \(\gamma\) be deterministic placement and delivery maps.
Placement uses a random seed \(Q_{\rm loc}\) independent of the local files:
\[
 (C_1,\ldots,C_h,G)=\Phi(A_{[N]},Q_{\rm loc}),
 \qquad X^{\mathbf e}=\gamma(G,\mathbf e).
\]
Let \(M_{\rm rx},M_{\rm snd},R_{\rm loc}\ge0\) be fixed bounds on the
receiver memory, sender memory, and broadcast load, respectively, normalized
by the local file entropy \(B\).  For every \(j\in[h]\) and demand
\(\mathbf e\), the bounds require
\[
H(C_j)\le M_{\rm rx}B,\quad H(G)\le M_{\rm snd}B,\quad
H(X^{\mathbf e})\le R_{\rm loc}B.
\]
The sender state recovers the local library:
\begin{equation}
H(A_{[N]}\given G)=0.
\label{eq:tifs-library-recoverable}
\end{equation}
For every receiver \(j\in[h]\) and demand \(\mathbf e\), the local view
\(V_j=(C_j,X^{\mathbf e})\) satisfies
\begin{align}
H(A_{e_j}\given V_j,\mathbf e)&=0,
\label{eq:tifs-local-dec}\\
I(A_{[N]\setminus\{e_j\}};V_j\given A_{e_j},\mathbf e)&=0.
\label{eq:tifs-composable-local-sec}
\end{align}
For an instance centered at physical sender \(k\), fix a bijection
\(\iota_k:[h]\to\cN_r(k)\setminus\{k\}\).  Local receiver \(j\) is
physical user \(\iota_k(j)\).  The local demand is
\[
\mathbf e=\bd^{(k)}:=(d_{\iota_k(1)},\ldots,d_{\iota_k(h)}),
\qquad e_j=d_{\iota_k(j)}.
\]
The seeds of distinct sender instances are mutually independent and independent
of the files and the outer-sharing randomness.  The state \(G\) includes every
variable used by the sender during delivery: the local-library representation,
delivery keys, and any retained placement randomness.  All of this state is
charged to \(M_{\rm snd}\).
\end{definition}

\subsection{Secure Lifting and the PCC Family}

\begin{theorem}[Secure lifting with complete sender-state accounting]
\label{thm:tifs-secure-lift}
Under the model of Section~\ref{sec:system-model}, let \(K\ge h+1\) and
\(h=2r\).  A local kernel satisfying
Definition~\ref{def:tifs-composable-kernel}, with parameters
\((M_{\rm rx},M_{\rm snd},R_{\rm loc})\), gives a file-secure D2D scheme at
\begin{equation}
    M=M_{\rm rx}+\frac{M_{\rm snd}}{h},
    \qquad R=\frac{K}{h}R_{\rm loc}.
\label{eq:tifs-secure-lift}
\end{equation}
\end{theorem}

Each file is shared across sender positions, with share entropy \(F/h\).
A user obtains one share from its own sender state and \(h\) shares from its
neighboring instances.  These \(h+1\) shares recover its requested file.
Section~\ref{sec:tifs-secure-lifting-framework} proves that its full observation
across the correlated instances reveals no information about nonrequested files.

\begin{theorem}[PCC-lift secure D2D family]
\label{thm:tifs-pcc-lift}
Under the model of Section~\ref{sec:system-model}, assume \(K\ge h+1\).
For \(t=0,\ldots,h-2\), the secure radius-\(r\) D2D network achieves
\begin{align}
M_t^{{\rm PCC}\to r}
&=\frac{N(t+1)}{h-t}+1+\frac{1}{t+1}, \nonumber\\
R_t^{{\rm PCC}\to r}&=\frac{K}{t+1}.
\label{eq:tifs-pcc-points}
\end{align}
It also achieves the high-memory endpoint
\begin{equation}
    (M,R)=\left(Nh,\frac{K}{h}\right).
\label{eq:tifs-pcc-endpoint}
\end{equation}
\end{theorem}

The memory in \eqref{eq:tifs-pcc-points} includes the PCC receiver placement,
the sender's full local sharing representation, and all delivery keys.
The factor \(K/h\) in the load follows from \(K\) local broadcasts, each
normalized by the share entropy \(F/h\).

\subsection{Converse and Optimality Guarantees}

\begin{theorem}[Radius-aware block-boundary converse]
\label{thm:tifs-secure-converse}
Under the model of Section~\ref{sec:system-model}, assume
\(N\ge K\ge h+1\).  For a block of \(s\in\{1,\ldots,K-1\}\)
consecutive users, let \(q_s\) be the number of demand instances used to
recover \(sq_s\) distinct files.  Let \(\delta_s\) count the external
senders audible to the block, and let \(\eta_s\) be the largest number of
those senders heard by one user in the block.  These quantities are
\[
    q_s:=\left\lfloor\frac Ns\right\rfloor,\qquad
    \delta_s:=\min\{h,K-s\},
\]
\[
    \eta_s:=\min\{K-s,\max\{h/2,h+1-s\}\}.
\]
For every \(s\in[K-1]\) satisfying \(q_s\delta_s>\eta_s\),
\begin{equation}
R_{\rm sec}^\star(M,r)
\ge
\frac{K\left[(sq_s-1)-(s-1)M\right]^+}
     {q_s\delta_s-\eta_s}.
\label{eq:tifs-secure-converse}
\end{equation}
When \(q_s\delta_s=\eta_s\), every feasible scheme satisfies
\((sq_s-1)\le(s-1)M\).
\end{theorem}

Choose an anchor user that hears \(\eta_s\) external senders.  File secrecy
removes the anchor's complete observation from the entropy bound.  Summing
the remaining boundary messages over cyclic shifts gives
\eqref{eq:tifs-secure-converse}.  The proof is in
Appendix~\ref{app:tifs-secure-converse}.

\begin{corollary}[Globally optimal PCC-lift endpoint]
\label{cor:tifs-pcc-endpoint}
Under the model of Section~\ref{sec:system-model}, if \(K\ge h+1\) and
\(N>1\), the point \((Nh,K/h)\) in \eqref{eq:tifs-pcc-endpoint} is exactly
optimal.
\end{corollary}

\begin{IEEEproof}
The singleton-cut argument in Section~\ref{sec:tifs-high-memory}, applied
to the \(N\) all-equal demand vectors, gives
\(R_{\rm sec}^\star(M,r)\ge K/h\) for every cache budget \(M\).
Theorem~\ref{thm:tifs-pcc-lift} attains this bound at \(M=Nh\).
\end{IEEEproof}

\begin{corollary}[Globally optimal secure high-memory ray]
\label{cor:tifs-secure-high-ray}
Under the same model, if \(K\ge h+1\) and \(N>1\), then
\[
R_{\rm sec}^\star(M,r)=\frac Kh,\qquad M\ge Nh.
\]
\end{corollary}

\begin{IEEEproof}
The lower bound \(K/h\) holds at every memory budget.  The placement attaining
it at \(Nh\) is feasible for every larger budget.
\end{IEEEproof}

\begin{theorem}[Uniform factor-\(15\) guarantee for PCC-lift]
\label{thm:tifs-pcc-gap}
Under the model of Section~\ref{sec:system-model}, assume
\(N\ge K\ge h+1\), where \(h=2r\).  Let
\(R_{{\rm PCC}\to r}^{\rm up}(M)\) be the lower convex envelope of
\eqref{eq:tifs-pcc-points} and \eqref{eq:tifs-pcc-endpoint}.  Re-index the
points in \eqref{eq:tifs-pcc-points} by \(u=t+1\), so
\(u\in\{1,\ldots,h-1\}\).  Their coordinates are
\begin{equation}
M_u^{{\rm PCC}\to r}
=\frac{Nu}{h-u+1}+1+\frac1u,\qquad
R_u^{{\rm PCC}\to r}=\frac Ku.
\label{eq:tifs-pcc-u-points}
\end{equation}
The endpoint has index \(u=h\) and coordinates
\[
M_h^{{\rm PCC}\to r}=Nh,\qquad R_h^{{\rm PCC}\to r}=\frac Kh.
\]
For
\[
M\in[M_1^{{\rm PCC}\to r},M_h^{{\rm PCC}\to r}]
=\left[\frac Nh+2,Nh\right],
\]
\begin{equation}
R_{{\rm PCC}\to r}^{\rm up}(M)
\le15\,R_{\rm sec}^\star(M,r).
\label{eq:tifs-pcc-gap}
\end{equation}
\end{theorem}

Appendix~\ref{app:tifs-pcc-gap} proves the bound by choosing a block size for
each PCC segment in \eqref{eq:tifs-secure-converse}.  The choice \(s=1\)
controls the high-memory segments through the lower bound \(K/h\).

\section{Composable Secure Lifting from Local Shares to Global Files}
\label{sec:tifs-secure-lifting-framework}

The outer sharing assigns a private local library to each sender.  Each
receiver obtains one requested share from every neighboring local instance.
The proof below combines these local views with the receiver's own sender state to
establish secrecy of the full physical observation.

\subsection{One-Private Outer Sharing}

The outer code must hide each file from any one sender and allow recovery
from the \(h+1\) shares available to a user.  For each file \(W_n\), generate
\(K\) shares \(S_{n,1},\ldots,S_{n,K}\), each with entropy \(B:=F/h\), such that
\begin{align}
H(W_n\given S_{n,\mathcal A})&=0,
&\mathcal A&\subseteq[K],\quad \abs{\mathcal A}=h+1,
\label{eq:tifs-outer-reconstruct}\\
I(W_n;S_{n,k})&=0,
&k&\in[K].
\label{eq:tifs-outer-private}
\end{align}
Each share is uniform on its alphabet.  Sharing is performed independently
across files, so the library \(\{S_{n,k}:n\in[N]\}\) at any one sender position
is independent of the complete file vector.

To construct the shares, choose a finite field \(\mathbb F\) with \(K\)
distinct nonzero points \(\alpha_1,\ldots,\alpha_K\).  Represent \(W_n\) by
\(h\) independent uniform field symbols \(W_{n,1},\ldots,W_{n,h}\).  Draw an
independent uniform pad \(U_n\in\mathbb F\) and define the polynomial
\(p_n(x)\), for \(x\in\mathbb F\), and its shares by
\[
    p_n(x)=U_n+\sum_{j=1}^{h}W_{n,j}x^j,\qquad
    S_{n,k}=p_n(\alpha_k).
\]
The pads are independent across files and independent of the library.  Any
\(h+1\) evaluations interpolate \(p_n\), while the uniform constant term makes
each evaluation independent of \(W_n\).  A share contains one field symbol
and a file contains \(h\), giving \(B=F/h\).

\begin{remark}[Finite-alphabet implementation]
\label{rem:tifs-finite-alphabet}
Let the local kernel have alphabet size \(a=2^{B_0}\ge2\), where \(B_0\)
is its local file entropy.  Write the prime factorization as
\(a=\prod_i p_i^{e_i}\), with distinct primes \(p_i\) and positive integer
exponents \(e_i\).  Choose an integer \(m\ge1\) with
\(p_i^{me_i}>K\) for every \(i\).
Run the polynomial sharing independently over each field
\(\mathbb F_{p_i^{me_i}}\).  A share in their product has alphabet size
\(a^m\), which is in bijection with the input alphabet of \(m\) independent
copies of the original kernel.  The resulting share entropy is \(B=mB_0\)
and the file entropy is \(F=hmB_0\).  Reconstruction and one-share privacy
hold in every field component, and independent kernel copies preserve the
normalized cache and load parameters.  This gives a finite implementation
for every finite uniform alphabet.
\end{remark}

For each sender \(k\) and file-index set \(\mathcal A\subseteq[N]\),
denote the complete local library and its subcollection by
\[
    \cS_k:=\{S_{n,k}:n\in[N]\},\qquad
    \cS_{k,\mathcal A}:=\{S_{n,k}:n\in\mathcal A\}.
\]
Sender \(k\) uses \(\cS_k\) as the local library for the \(h\) receivers in
\(\cN_r(k)\setminus\{k\}\).

\subsection{Requested-Share Channels and Multi-Sender Composition}

At sender \(k\), apply a local scheme satisfying
Definition~\ref{def:tifs-composable-kernel} to the library
\(A_n=S_{n,k}\) and demand \(\bd^{(k)}\).  For a physical receiver
\(u\in\cN_r(k)\setminus\{k\}\), write \(C_u^{(k)}\) for the cache of
local receiver \(\iota_k^{-1}(u)\).  Denote sender \(k\)'s broadcast by
\(X_k\), suppressing its dependence on the fixed demand, and define the
receiver's local view as
\[
V_u^{(k)}:=(C_u^{(k)},X_k).
\]
Local recovery gives \(S_{d_u,k}\).  Conditioned local secrecy makes the law
of \(V_u^{(k)}\) depend on the local library only through that requested
share.  The next lemma composes these channels across the audible senders.

\begin{lemma}[Multi-sender secrecy composition]
\label{lem:tifs-multi-sender}
Fix \(u,\bd\), and let \(\mathcal K_u:=\cN_r(u)\setminus\{u\}\).
Generate the outer sharings independently across files, with uniform
one-private shares.  The local seeds are mutually independent and independent
of the files and all outer-sharing randomness.  Each local kernel satisfies
\eqref{eq:tifs-composable-local-sec}, and \(G_u\) is generated from
\(\cS_u\) and its local seed.  Then
\begin{equation}
\begin{aligned}
&I\!\left(W_{[N]\setminus\{d_u\}};
G_u,\{S_{n,u}:n\in[N]\},\right.\\
&\hspace{2.6em}\left.
\{V_u^{(k)}:k\in\mathcal K_u\}
\given W_{d_u},\bd\right)=0.
\end{aligned}
\label{eq:tifs-multi-sender}
\end{equation}
\end{lemma}

\begin{IEEEproof}
We separate the user's home state and requested shares from the neighboring
nonrequested shares.  Let \(\bar W_u=W_{[N]\setminus\{d_u\}}\) be the
nonrequested file vector and \(\mathbf S_k=(S_{1,k},\ldots,S_{N,k})\) the
local library at sender \(k\).  Write \(H_u\) for the home state,
\(D_u\) for the requested neighboring shares, and \(A_u\) for the other
neighboring shares.  Let \(\mathbf V_u\) collect the neighboring views:
\begin{align*}
 H_u&:=\bigl(G_u,\mathbf S_u\bigr),\\
 D_u&:=\bigl(S_{d_u,k}:k\in\mathcal K_u\bigr),\\
 A_u&:=\bigl(S_{n,k}:n\ne d_u,\ k\in\mathcal K_u\bigr),\\
 \mathbf V_u&:=\bigl(V_u^{(k)}:k\in\mathcal K_u\bigr).
\end{align*}
In the probability expressions below, write
\(h_0\), \(\mathbf s=(s_k:k\in\mathcal K_u)\), \(a\), and
\(\mathbf v=(v_k:k\in\mathcal K_u)\) for their respective realizations.
The nonrequested home shares \((S_{n,u}:n\ne d_u)\) are independent of
\(\bar W_u\) by one-share privacy and independent of all requested-file
shares by filewise generation.  Since \(G_u\) is a function of
\(\mathbf S_u\) and an independent seed, this gives
\begin{equation}
I(\bar W_u;H_u,D_u\given W_{d_u},\bd)=0.
\label{eq:tifs-home-requested-independent}
\end{equation}

For each \(k\in\mathcal K_u\), conditioned local secrecy gives
\begin{equation}
V_u^{(k)}
\;--\;\bigl(S_{d_u,k},\bd^{(k)}\bigr)
\;--\;\bigl(S_{n,k}:n\ne d_u\bigr).
\label{eq:tifs-local-markov}
\end{equation}
The local library is a vector of independent uniform shares, so it has full
support.  Hence a conditional probability distribution \(Q_k\) describes
the local view using the requested share alone.  For every local-library
realization \(\boldsymbol\sigma_k=(\sigma_{1,k},\ldots,\sigma_{N,k})\),
\begin{equation}
P_{V_u^{(k)}\given\mathbf S_k,\bd}
(v_k\given\boldsymbol\sigma_k,\bd)
=Q_k(v_k\given\sigma_{d_u,k},\bd^{(k)}).
\label{eq:tifs-local-channel}
\end{equation}
The local demand \(\bd^{(k)}\) is a fixed projection of \(\bd\).

Given the outer shares, each neighboring view is generated by its sender's
independent seed.  Conditioning additionally on \(W_{[N]}\) and \(H_u\)
preserves this independence, because \(H_u\) uses the independent home seed.
Consequently, at supported realizations with \(D_u=\mathbf s\),
\begin{equation}
\begin{aligned}
&P(\mathbf V_u=\mathbf v\given W_{[N]},H_u,D_u=\mathbf s,A_u,\bd)\\
&\qquad=\prod_{k\in\mathcal K_u}
Q_k(v_k\given s_k,\bd^{(k)}).
\end{aligned}
\label{eq:tifs-seed-product}
\end{equation}
The product does not depend on \(A_u\).  Summing over its realizations \(a\)
gives
\begin{align}
&P_{H_u,D_u,\mathbf V_u\given W_{d_u},\bar W_u,\bd}
(h_0,\mathbf s,\mathbf v)\nonumber\\
&=\sum_a
P_{H_u,D_u,A_u\given W_{d_u},\bar W_u,\bd}(h_0,\mathbf s,a)
\prod_{k\in\mathcal K_u}Q_k(v_k\given s_k,\bd^{(k)})\nonumber\\
&=P_{H_u,D_u\given W_{d_u},\bar W_u,\bd}(h_0,\mathbf s)
\prod_{k\in\mathcal K_u}Q_k(v_k\given s_k,\bd^{(k)})\nonumber\\
&=P_{H_u,D_u\given W_{d_u},\bd}(h_0,\mathbf s)
\prod_{k\in\mathcal K_u}Q_k(v_k\given s_k,\bd^{(k)}).
\label{eq:tifs-multi-sender-factorization}
\end{align}
The last equality uses \eqref{eq:tifs-home-requested-independent}.
The resulting law is independent of \(\bar W_u\), proving
\eqref{eq:tifs-multi-sender}.
\end{IEEEproof}

\subsection{Secure Lifting with Complete Sender-State Accounting}

\begin{IEEEproof}[Proof of Theorem~\ref{thm:tifs-secure-lift}]
Use the outer sharing above, with the alphabet matching in
Remark~\ref{rem:tifs-finite-alphabet} when needed.  User \(u\) stores
\[
Z_u=\bigl(G_u,\{C_u^{(k)}:k\in\mathcal K_u\}\bigr),
\qquad \mathcal K_u=\cN_r(u)\setminus\{u\}.
\]
Since \(\abs{\mathcal K_u}=h\), entropy subadditivity and the kernel budgets give
\begin{align*}
H(Z_u)&\le H(G_u)+\sum_{k\in\mathcal K_u}H(C_u^{(k)})\\
&\le(M_{\rm snd}+hM_{\rm rx})B,\\
\sum_{k=1}^K H(X_k^\bd)&\le KR_{\rm loc}B.
\end{align*}
Dividing by \(F=hB\) gives the achievable budgets in
\eqref{eq:tifs-secure-lift}.

User \(u\) recovers \(S_{d_u,k}\) from each neighboring view through
\eqref{eq:tifs-local-dec} and recovers \(S_{d_u,u}\) from \(G_u\).
These \(h+1\) distinct outer shares reconstruct \(W_{d_u}\) by
\eqref{eq:tifs-outer-reconstruct}.

By Lemma~\ref{lem:tifs-multi-sender}, \(G_u\) and all neighboring views
jointly reveal no information about the nonrequested files, conditioned on
the requested file and the public demand.
The self-transmission \(X_u^\bd=\gamma(G_u,\bd^{(u)})\) is a deterministic
function of \(G_u\) and the public demand.  Thus the user's complete
observation satisfies \eqref{eq:tifs-file-sec-conditioned}, which is equivalent
to \eqref{eq:tifs-file-sec}.
\end{IEEEproof}

\begin{example}[Radius-two secure lifting at \(t=0\)]
\label{ex:tifs-secure-lift}
Take \(K=6\), \(r=2\), \(h=4\), and \(N=3\).  Generate six one-private shares
of each file, with share entropy \(F/4\), so any five reconstruct the file.
Sender \(k\) stores its three-share library \(\{S_{n,k}:n\in[3]\}\) and a
uniform key \(P_{ku}\) of entropy \(F/4\) for every
\(u\in\mathcal K_k:=\cN_r(k)\setminus\{k\}\).  All keys are mutually
independent and independent of the files and outer shares.  Receiver \(u\)
stores \(P_{ku}\) for each neighboring sender \(k\).  Under
\(\bd=(1,2,3,1,2,3)\), sender \(k\) broadcasts
\[
X_{ku}=P_{ku}+S_{d_u,k},\qquad u\in\mathcal K_k,
\]
over the share field.

User \(1\) hears senders \(2,3,5,6\).  Its four receiver keys recover
\(S_{1,2},S_{1,3},S_{1,5},S_{1,6}\), and its sender state supplies
\(S_{1,1}\), completing reconstruction of \(W_1\).  Within each local
instance, packets intended for other receivers are masked by independent
keys absent from the receiver's cache.  The full local view therefore
satisfies \eqref{eq:tifs-composable-local-sec}.
Lemma~\ref{lem:tifs-multi-sender} gives secrecy of the combined observation.
The local parameters are
\[
M_{\rm rx}=1,\qquad M_{\rm snd}=N+h=7,\qquad R_{\rm loc}=h=4.
\]
Equation~\eqref{eq:tifs-secure-lift} then gives \(M=11/4\) and \(R=6\).
\end{example}

\section{Secure PCC-Lift with Sender-State Accounting}
\label{sec:tifs-pcc-ach}

\subsection{PCC Placement and Delivery}

For local delivery, we use PCC, which combines ramp secret sharing with
keyed coded multicast
\cite{RavindrakumarPandaKaramchandaniPrabhakaran2018}.  We account for the
receiver cache, the sender's stored shares and keys, and the broadcast load.
For an integer
\(t\in\{0,\ldots,h-2\}\), its receiver memory and load at sender \(k\) are
\begin{equation}
M_{\rm rx}^{\rm PCC}(t)=1+\frac{Nt}{h-t},
\qquad
R_{\rm loc}^{\rm PCC}(t)=\frac{h}{t+1},
\label{eq:tifs-pcc-local}
\end{equation}
with memory and load normalized by the local file entropy \(B=F/h\).
Receiver shares occupy \(Nt/(h-t)\) local file units, and keys
occupy one unit.  All encoding and delivery operations are linear over a
binary extension field \(\mathbb F=\mathbb F_{2^m}\), with a positive integer
\(m\) chosen large enough for
the outer sharing and the local ramp code.

Fix sender \(k\).  Use the ordering \(\iota_k\) from
Definition~\ref{def:tifs-composable-kernel} to label its receivers by
\(j\in[h]\).  Write \(\mathbf e=\bd^{(k)}\), so
\(e_j=d_{\iota_k(j)}\), and denote receiver \(j\)'s cache by
\(\widetilde C_j^{(k)}:=C_{\iota_k(j)}^{(k)}\).
Let \(n_t\) be the number of coded shares of each local file, \(z_t\) the
number stored by each receiver and the number of random pad blocks, and
\(b_t\) the number of source blocks in that file.  At \(t=0\), set
\(z_0=0\).  The counts are
\begin{align*}
 n_t&:=\binom ht, & z_t&:=\binom{h-1}{t-1},\\
 b_t&:=n_t-z_t=\binom{h-1}{t}.
\end{align*}
Take \(2^m>\max\{K,n_t\}\).  Represent each local file
\(S_{n,k}\) by \(b_t\) equal-length field-symbol blocks and draw \(z_t\)
uniform pad blocks.  All pad blocks are mutually independent across block
indices and sender--file pairs, and independent of the files and
outer-sharing randomness.
Let \(\ell_t\) be the entropy in bits of each coded share.  A linear
\((z_t,n_t)\) ramp code produces one share \(S_{n,k,T}\) for each
\(t\)-element receiver subset \(T\subseteq[h]\):
\[
    S_{n,k,T},\qquad T\subseteq[h],\quad |T|=t,
\]
with
\begin{equation}
    \ell_t=\frac{B}{b_t}
    =\frac{hB}{(h-t)\binom ht}
    =\frac{F}{h\binom{h-1}{t}}.
\label{eq:tifs-pcc-symbol-size}
\end{equation}
Every set of at most \(z_t\) shares is independent of \(S_{n,k}\), and all
\(n_t\) shares reconstruct it.  Receiver \(j\) stores the shares indexed by
sets \(T\) containing \(j\), exactly \(z_t\) shares per local file.
The cached shares reveal no local file content and occupy
\[
\binom{h-1}{t-1}\ell_t=\frac{tB}{h-t}
\]
bits per file, giving \(Nt/(h-t)\) local file units over the library.  At
\(t=0\), the receiver caches no ramp shares.  Taking \(F\) to be an integer
multiple of \(mhb_t\) makes each packet an integer number of field symbols.
Independent vector extensions preserve the normalized point.

For every sender \(k\) and \((t+1)\)-subset \(\mathcal V\subseteq[h]\), draw
a uniform key \(K_{k,\mathcal V}\) of size \(\ell_t\).  These keys are
mutually independent over \((k,\mathcal V)\) and independent of the files,
outer-sharing randomness, and local ramp pads.  Sender \(k\) and all receivers
in \(\mathcal V\) store \(K_{k,\mathcal V}\).  Each receiver stores
\(\binom{h-1}{t}\) keys with aggregate size
\[
\binom{h-1}{t}\ell_t=B,
\]
so its keys occupy one local file unit.  Sender \(k\) broadcasts one packet
for every \(|\mathcal V|=t+1\):
\[
    K_{k,\mathcal V}
    +\sum_{j\in \mathcal V} S_{e_j,k,\mathcal V\setminus\{j\}} .
\]
For receiver \(j\in\mathcal V\), the cache contains the key and every summand
indexed by \(i\ne j\), since \(j\in\mathcal V\setminus\{i\}\).  The packet
therefore supplies the missing share
\(S_{e_j,k,\mathcal V\setminus\{j\}}\).  Across the subsets containing \(j\),
these packets recover all requested shares whose indices omit \(j\).
Together with its cached shares, receiver \(j\) recovers \(S_{e_j,k}\).
The \(\binom h{t+1}\) packets have aggregate size
\(\binom h{t+1}\ell_t=hB/(t+1)\), giving the local rate in
\eqref{eq:tifs-pcc-local}.  Their independent one-time-pad keys also make
the complete delivery vector independent of the local library.

The local sender state contains the complete ramp-coded representation of all
\(N\) local files and every delivery key.  Its normalized size is
\begin{equation}
M_{\rm snd}^{\rm PCC}(t)=\frac{Nh}{h-t}+\frac{h}{t+1}.
\label{eq:tifs-pcc-sender}
\end{equation}
The \(N\binom ht\) shares store the local library and its ramp pads.
The remaining \(\binom h{t+1}\) blocks store the delivery keys.
After lifting, this sender state occupies
\(M_{\rm snd}^{\rm PCC}(t)/h\) original file units at each physical user.

\begin{remark}[Source and encoded packet counts]
The source-symbol subpacketization \(L_t^{\rm src}\) counts the
\(\ell_t\)-bit field-symbol blocks in one original file:
\[
L_t^{\rm src}:=\frac{F}{\ell_t}
=h\binom{h-1}{t}.
\]
Encoding generates \(K\binom ht\) coded blocks across the \(K\) sender
positions.
\end{remark}
\begin{IEEEproof}[Proof of Theorem~\ref{thm:tifs-pcc-lift}]
Substituting \eqref{eq:tifs-pcc-local} and \eqref{eq:tifs-pcc-sender} into
\eqref{eq:tifs-secure-lift} gives
\begin{align*}
M
&=1+\frac{Nt}{h-t}
  +\frac1h\left(\frac{Nh}{h-t}+\frac{h}{t+1}\right)\\
&=\frac{N(t+1)}{h-t}+1+\frac1{t+1},
\end{align*}
and \(R=(K/h)(h/(t+1))=K/(t+1)\).

The packet decoder gives \eqref{eq:tifs-local-dec}.  To prove conditioned
local secrecy, fix local receiver \(j\) and let
\(A=\cS_{k,[N]\setminus\{e_j\}}\) collect its nonrequested local files.
For every such file, the shares in \(\widetilde C_j^{(k)}\) form an
unauthorized ramp-code set.  Encodings are
independent across files, and the requested-file shares and keys use independent
randomness.  Hence
\[
 I(A;\widetilde C_j^{(k)}\given S_{e_j,k},\mathbf e)=0.
\]
Partition the delivery packets according to whether \(j\in\mathcal V\).
When \(j\in\mathcal V\), the receiver's cache contains the key and all
summands except \(S_{e_j,k,\mathcal V\setminus\{j\}}\).
Every summand involving a nonrequested file is therefore cached.  Given
the cache and \(S_{e_j,k}\), packets in this group reveal requested-file
shares whose remaining randomness is independent of \(A\).
When \(j\notin\mathcal V\), the packet is masked by the independent
one-time-pad key \(K_{k,\mathcal V}\), which is absent from
\(\widetilde C_j^{(k)}\).  The distinct keys for this second group make its
packets jointly independent of the local library, cache, and first packet
group.  Combining the two groups gives
\[
 I(A;\widetilde C_j^{(k)},X_k\given S_{e_j,k},\mathbf e)=0,
\]
which gives \eqref{eq:tifs-composable-local-sec}.  The complete ramp-coded
library in \(G_k\) determines \(\cS_k\) and every transmitted share.  The
keys in \(G_k\) determine all packet masks.  Thus the sender recovers its
local library and generates every transmission from \(G_k\).
The sender-state size is given by
\eqref{eq:tifs-pcc-sender}.
The independent sender seeds satisfy the remaining condition of
Theorem~\ref{thm:tifs-secure-lift}.

\emph{High-memory endpoint.}
To attain global load \(K/h\), we construct a local scheme with one
\(B\)-bit broadcast, so \(R_{\rm loc}=1\).  Each receiver will cache
\(h-1\) shares of every local file, and the broadcast will supply its missing
requested share.
Choose \(a\in\mathbb F_{2^m}\setminus\{0,1\}\), using the same field as the
outer sharing.  Such an \(a\) exists because \(2^m>K\ge h+1\ge3\).
For every sender--file pair \((k,n)\), draw \(h-1\) uniform \(B\)-bit
field-symbol blocks \(U_{n,k,1},\ldots,U_{n,k,h-1}\).  All pad blocks are
mutually independent and independent of the files and outer-sharing
randomness.  Define the endpoint shares by
\[
\begin{aligned}
T_{n,k,i}&=U_{n,k,i}, &&i<h,\\
T_{n,k,h}&=S_{n,k}+a\sum_{i=1}^{h-1}U_{n,k,i}.
\end{aligned}
\]
The \(h\) shares recover \(S_{n,k}\) by
\[
S_{n,k}=T_{n,k,h}-a\sum_{i=1}^{h-1}T_{n,k,i}.
\]
Any \(h-1\) shares are independent of \(S_{n,k}\): omitting the last share
leaves only independent pads, and omitting share \(i<h\) leaves the unknown
uniform term \(aU_{n,k,i}\) in the last share.  Since \(a\ne0\), the file
and any \(h-1\) shares also determine the remaining share.

Local receiver \(j\) stores
\(\{T_{n,k,i}:n\in[N],\,i\in[h]\setminus\{j\}\}\), giving
\(M_{\rm rx}=N(h-1)\).

For the local demand \(\mathbf e=(e_1,\ldots,e_h)\), sender \(k\) broadcasts
\[
X_k=\sum_{j=1}^h T_{e_j,k,j}.
\]
Receiver \(j\) knows every summand indexed by \(i\ne j\), since it caches
the \(i\)-th share of every local file.  It recovers \(T_{e_j,k,j}\) and then
\(S_{e_j,k}\).  This decoder works for every local demand vector, including
repeated requests, and the \(B\)-bit packet gives \(R_{\rm loc}=1\).

Given \(S_{e_j,k}\), the \(h-1\) cached requested-file shares determine
\(T_{e_j,k,j}\).  Hence \(X_k\) is a deterministic function of
\((\widetilde C_j^{(k)},S_{e_j,k},\mathbf e)\), and
\[
\begin{aligned}
&I(\cS_{k,[N]\setminus\{e_j\}};
\widetilde C_j^{(k)},X_k\given S_{e_j,k},\mathbf e) \\
&\quad=
I(\cS_{k,[N]\setminus\{e_j\}};
\widetilde C_j^{(k)}\given S_{e_j,k},\mathbf e).
\end{aligned}
\]
For every \(n\ne e_j\), the cache contains exactly \(h-1\) shares of
\(S_{n,k}\), an unauthorized set independent of that local file.  Independent
sharing across file indices makes the full nonrequested cache contribution
independent of \(\cS_{k,[N]\setminus\{e_j\}}\), so the right-hand side is zero.
This proves \eqref{eq:tifs-composable-local-sec}.

The sender stores the complete \(h\)-share representation of its local
library, of size \(NhB=NF\), and hence \(M_{\rm snd}=Nh\).
Theorem~\ref{thm:tifs-secure-lift} gives
\((M,R)=(N(h-1)+N,K/h)=(Nh,K/h)\).
\end{IEEEproof}
\begin{proposition}[Global delivery-only secrecy of the high-memory endpoint]
\label{prop:tifs-endpoint-delivery-secrecy}
For the high-memory endpoint constructed above, every global demand \(\bd\)
gives
\begin{equation}
I\!\left(W_{[N]};X_1^{\bd},\ldots,X_K^{\bd}\given\bd\right)=0.
\label{eq:tifs-endpoint-global-delivery-secrecy}
\end{equation}
\end{proposition}
\begin{IEEEproof}
Fix sender \(k\) and its projected demand \(\mathbf e=\bd^{(k)}\).
For each file index \(n\), let \(\mathcal A_{k,n}\subseteq[h]\) be the
set of local receivers requesting that file.  Let \(Y_{k,n}\) be the sum
of the corresponding endpoint shares, with an empty sum equal to zero:
\[
\mathcal A_{k,n}:=\{j\in[h]:e_j=n\},\qquad
Y_{k,n}:=\sum_{j\in\mathcal A_{k,n}}T_{n,k,j}.
\]
For \(1\le|\mathcal A_{k,n}|\le h-1\), the selected shares form an unauthorized
set, so \(Y_{k,n}\) is independent of \(S_{n,k}\).  When all \(h\) receivers
request file \(n\),
\[
Y_{k,n}=S_{n,k}+(a+1)\sum_{i=1}^{h-1}U_{n,k,i}.
\]
The coefficient \(a+1\) is nonzero, so this sum is also independent of
\(S_{n,k}\).  The empty-set case gives \(Y_{k,n}=0\).  These cases cover
every request multiplicity.  Independent pads across files give
\begin{equation}
P_{(Y_{k,n})_{n\in[N]}\mid\cS_k,\bd}
= P_{(Y_{k,n})_{n\in[N]}\mid\bd}.
\label{eq:tifs-one-sender-endpoint-delivery-channel}
\end{equation}
Since \(X_k^{\bd}=\sum_nY_{k,n}\),
\(P_{X_k^{\bd}\mid\cS_k,\bd}=P_{X_k^{\bd}\mid\bd}\).

Given the outer libraries, independent pads across senders make their
broadcasts conditionally independent.  We average over the outer libraries
to obtain the joint broadcast law.  Write \(w_{[N]}\), \(s_{[K]}\), and
\(x_{[K]}\) for realizations of the file vector, the tuple of outer libraries,
and the broadcast vector, respectively.  Then
\begin{align*}
&P_{X_{[K]}^{\bd}\mid W_{[N]},\bd}
   (x_{[K]}\mid w_{[N]},\bd)\\
&\quad=\sum_{s_{[K]}}P_{\cS_{[K]}\mid W_{[N]}}
  (s_{[K]}\mid w_{[N]})
  \prod_{k=1}^K P_{X_k^{\bd}\mid\cS_k,\bd}(x_k\mid s_k,\bd)\\
&\quad=\prod_{k=1}^K P_{X_k^{\bd}\mid\bd}(x_k\mid\bd),
\end{align*}
where the second equality applies
\eqref{eq:tifs-one-sender-endpoint-delivery-channel} before summing over the
correlated outer shares.  The resulting distribution is independent of
\(w_{[N]}\), proving
\eqref{eq:tifs-endpoint-global-delivery-secrecy}.
\end{IEEEproof}
\begin{example}[A coded secure PCC-lift instance]
\label{ex:tifs-secure-pcc-coded}
Take \(K=6\), \(r=2\), \(h=4\), \(N=4\), and \(t=1\).
In this example, label receivers and share-index sets by physical user
indices in \(\mathcal K_k=\cN_r(k)\setminus\{k\}\).
Each outer share has size \(B=F/4\).  Sender \(k\)'s local ramp code produces
four shares per local file,
\[
 S_{n,k,\{u\}},\qquad u\in\mathcal K_k,
\]
of size \(B/3\).  Receiver \(u\) caches the share indexed by \(\{u\}\) for
every \(n\).  For every receiver pair in \(\mathcal K_k\), a key of size
\(B/3\) is stored by the sender and the two receivers.

Consider sender \(1\), whose local receivers are
\(\mathcal K_1=\{5,6,2,3\}\), under the demand
\(\bd=(1,3,4,2,1,2)\).  For every pair
\(\{u,v\}\subseteq\mathcal K_1\), it transmits
\[
X_{1;uv}=K_{1,\{u,v\}}+
S_{d_u,1,\{v\}}+ S_{d_v,1,\{u\}}.
\]
For example,
\(X_{1;56}=K_{1,\{5,6\}}+ S_{1,1,\{6\}}+ S_{2,1,\{5\}}\).
User \(5\) cancels its cached key \(K_{1,\{5,6\}}\) and share
\(S_{2,1,\{5\}}\) to obtain \(S_{1,1,\{6\}}\).
User \(6\) obtains \(S_{2,1,\{5\}}\) symmetrically.  The
pairs \(\{5,2\}\) and \(\{5,3\}\) supply user \(5\)'s other two missing
shares of \(S_{1,1}\), which combine with its cached share
\(S_{1,1,\{5\}}\).  Every packet serves two receivers.
Packets for pairs that exclude user \(5\) are masked by keys absent from
its cache.  Its complete local view
\((C_5^{(1)},X_1)\) therefore satisfies the conditioned secrecy relation
required by Lemma~\ref{lem:tifs-multi-sender}.

With the local operation applied at every sender, user \(5\) recovers the four
neighboring outer shares \(S_{1,k}\), \(k\in\{1,3,4,6\}\), and joins them with
its home share \(S_{1,5}\) to reconstruct \(W_1\).
Lemma~\ref{lem:tifs-multi-sender} gives file secrecy for user \(5\)'s
complete observation.  The cache sizes are
\[
 M_{\rm rx}=1+\frac{N}{h-1}=\frac73,\qquad
 M_{\rm snd}=\frac{Nh}{h-1}+\frac h2=\frac{22}{3},
\]
and Theorem~\ref{thm:tifs-secure-lift} gives
\[
 (M,R)=\left(\frac{25}{6},3\right),
\]
which is \eqref{eq:tifs-pcc-points} at \(t=1\).
\end{example}
\subsection{Improvements in the Local Kernel}

The lifting map \eqref{eq:tifs-secure-lift} is affine and increasing in each
local parameter.  Let \(A\) and \(B\) be two local schemes satisfying
Definition~\ref{def:tifs-composable-kernel}, with
\[
M_{\rm rx}^{A}\le M_{\rm rx}^{B},\qquad
M_{\rm snd}^{A}\le M_{\rm snd}^{B},\qquad
R_{\rm loc}^{A}\le R_{\rm loc}^{B}.
\]
Their lifted points satisfy \(M^{A}\le M^{B}\) and \(R^{A}\le R^{B}\).
A strict decrease in either memory coordinate gives \(M^{A}<M^{B}\).
A strict decrease in local rate gives \(R^{A}<R^{B}\).
Thus local improvements under the stated kernel conditions transfer directly
to the radius-\(r\) D2D tradeoff.

\section{Secure Converse and Performance Guarantees}
\label{sec:tifs-secure-converse-details}

\subsection{Radius-Aware Block-Boundary Converse}

The block-boundary converse bounds the load needed to deliver files to a
contiguous user block.  It uses demand vectors in which every user requests
a distinct file and aligns these demands across cyclic shifts of the block.
For each block, the argument selects an anchor user and conditions on that
user's requested file in one demand instance.  By file secrecy, the
anchor's cache and audible broadcasts reveal no information about the
other \(sq_s-1\) files requested by the block across its \(q_s\) demand
instances.  Removing the anchor's cache and the \(\eta_s\) boundary
broadcasts it hears in that instance leaves \(s-1\) caches and
\(q_s\delta_s-\eta_s\) boundary broadcasts in the entropy bound.
Summing this bound over cyclic shifts gives
Theorem~\ref{thm:tifs-secure-converse}.
\subsection{Exact High-Memory Ray and Uniform Factor-15 Guarantee}
\label{sec:tifs-high-memory}

Considering one user's recovery under all-equal demands gives the lower bound \(K/h\)
for every cache size under \(K\ge h+1\) and \(N>1\).  The PCC-lift endpoint
attains this bound at \(M=Nh\), so it is optimal at and above this cache size.  For
\(N\ge K\ge h+1\), the block-boundary converse also supplies the uniform
factor-\(15\) guarantee throughout \([N/h+2,Nh]\).
\begin{IEEEproof}[Proof of Corollary~\ref{cor:tifs-pcc-endpoint}]
For each \(n\in[N]\), let
\(\bd^{(n)}=(n,\ldots,n)\) be the all-equal demand vector.  Fix a user \(i\),
and let \(\cG_i^{(1)}\) be the \(h=2r\) other transmitters heard by that user.
The user's own transmission is a function of its cache and the demand, so
reliability over the \(N\) demand vectors gives
\[
H\!\left(W_{[N]}\given
Z_i,\{X_{\cG_i^{(1)}}^{\bd^{(n)}}:n\in[N]\}\right)=0.
\]
Choose \(\bd^{(1)}\) as the anchor demand and define
\[
A_i:=\bigl(Z_i,X_{\cG_i^{(1)}}^{\bd^{(1)}}\bigr),\qquad
B_i:=\{X_{\cG_i^{(1)}}^{\bd^{(n)}}:n=2,\ldots,N\}.
\]
The conditioned file-secrecy constraint for user \(i\) under
\(\bd^{(1)}\) gives
\[
I(W_{[N]\setminus\{1\}};A_i\given W_1)=0.
\]
The pair \((A_i,B_i)\) recovers the library, and the chain rule yields
\[
(N-1)F\le H(B_i)
\le \sum_{n=2}^N\sum_{k\in\cG_i^{(1)}}
H(X_k^{\bd^{(n)}}).
\]
Summing over all \(K\) singleton blocks counts every transmitter in exactly
\(h\) sets \(\cG_i^{(1)}\).  The peak-load constraint therefore gives
\[
K(N-1)F\le h\sum_{n=2}^N\sum_{k=1}^K
H(X_k^{\bd^{(n)}})\le h(N-1)RF.
\]
Since \(N>1\), division gives
\(R_{\rm sec}^\star(M,r)\ge K/h\).  The endpoint in
Theorem~\ref{thm:tifs-pcc-lift} attains this load.
\end{IEEEproof}
To prove the uniform guarantee, choose a converse block size for each
PCC-lift memory-sharing segment and compare the two load bounds on that
segment.  The segment-wise inequalities establish
\(R_{{\rm PCC}\to r}^{\rm up}(M)\le
15R_{\rm sec}^\star(M,r)\) over \([N/h+2,Nh]\).
\subsection{Secure Tradeoffs Across Connectivity Radii}

Under \(N\ge K\ge h+1\), define the maximized converse
\begin{equation}
R_{\rm lb}(M):=
\max_{\substack{s\in[K-1]:\\q_s\delta_s>\eta_s}}
\frac{K[(sq_s-1)-(s-1)M]^+}{q_s\delta_s-\eta_s}.
\label{eq:tifs-numerical-lb}
\end{equation}
The \(s=1\) term gives \(R_{\rm lb}(M)\ge K/h\).
Figure~\ref{fig:tifs-secure-numerical} compares this bound with the PCC-lift
envelope at \(N=K=18\) for \(r=2\) and \(r=6\).  Each curve is drawn from
its exact affine segments.  Because both curves are piecewise affine and
the converse is positive, their ratio is monotone or constant between
successive breakpoints.  Its maximum over the memory interval is therefore
obtained by checking the endpoints and the union of the two curves'
breakpoints.
\begin{figure*}[t]
  \centering
  \includegraphics[width=0.98\textwidth]{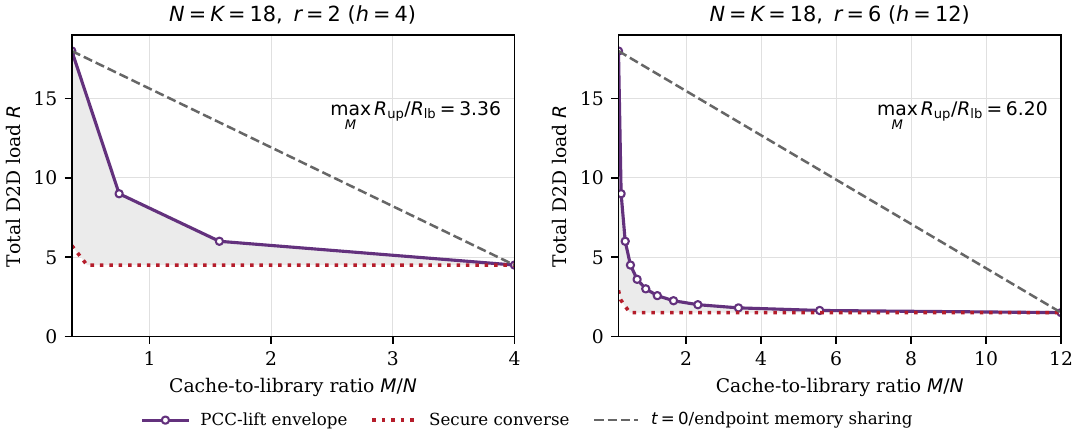}
  \caption{Secure memory--load tradeoffs at \((N,K)=(18,18)\).
  The purple curve is the PCC-lift lower convex envelope, and the red curve
  is \(R_{\rm lb}(M)\) in \eqref{eq:tifs-numerical-lb}.
  The gray line memory-shares between secure unicast (\(t=0\)) and the exact
  endpoint.  In the annotations, \(R_{\rm up}\) denotes the PCC-lift envelope.
  Each annotated value equals
  \(\max_{M\in[N/h+2,Nh]}
  R_{{\rm PCC}\to r}^{\rm up}(M)/R_{\rm lb}(M)\), rounded to two decimals.
  The exact maxima are \(47/14\) for \(r=2\) and \(31/5\) for \(r=6\).
  At \(M/N=h\), the PCC-lift load and \(R_{\rm lb}(M)\) both equal \(K/h\).}
  \label{fig:tifs-secure-numerical}
\end{figure*}
For \(r=2\), the largest ratio is \(47/14\) at \(M=35/4\), where the
\(s=1\) and \(s=2\) converse branches meet.  For \(r=6\), the largest ratio
is \(31/5\) at the left endpoint \(M=7/2\).
Increasing the radius from \(2\) to \(6\) lowers the high-memory optimum
from \(9/2\) to \(3/2\).  The corresponding PCC-lift endpoint operating
points have cache sizes \(72\) and \(216\), respectively.

\section{Modular MDS Lifting Benchmark}
\label{sec:nonsecure-counterpart}

We study the storage and communication needed for reliable file recovery
in the non-secure model of Section~\ref{sec:system-model}.
We use a shared-link scheme whose sender generates each broadcast from the
local library and public demand.  An MDS outer code then gives a
radius-limited D2D scheme by storing each local sender's library in its
physical cache.  The memory
budget also includes the receiver caches from neighboring instances.

Consider a local library \(A_{[N]}\) of independent, uniformly distributed
files of equal entropy.  Let the positive integer \(K_0\) be the number of
receivers in the shared-link kernel, and let \(M_0,R_0\ge0\) be its memory
and load budgets in local-file units.  The kernel
has memory--rate pair \((M_0,R_0)\) when each receiver cache has entropy at
most \(M_0\) file units, the broadcast has entropy at most \(R_0\) file units
for every demand, and every receiver decodes its requested file.
For a fixed public local demand
\(\mathbf e=(e_1,\ldots,e_{K_0})\in[N]^{K_0}\), the deterministic delivery
map \(\gamma\) produces
\[
X^{\mathbf e}=\gamma(A_{[N]},\mathbf e).
\]
We call the kernel sender-realizable when the local library is its complete
delivery-encoder state.  In the D2D construction, each physical sender stores
this state in its cache.
\subsection{MDS Lifting}

\begin{theorem}[MDS lifting of a shared-link kernel]
\label{thm:tifs-ns-lift}
Under the single-shot, non-interactive, noiseless orthogonal local-broadcast
model of Section~\ref{sec:system-model}, let \(h=2r\) and \(K\ge h+1\).  Every
sender-realizable \(K_0=h\) receiver shared-link kernel with memory--rate pair
\((M_0,R_0)\), normalized by the local MDS-share size and with the local library
as its complete encoder state, induces the radius-\(r\) D2D point
\begin{equation}
    (M,R)=
    \left(\frac{N+hM_0}{h+1},\frac{K}{h+1}R_0\right).
\label{eq:tifs-ns-lift}
\end{equation}
\end{theorem}

\begin{IEEEproof}[Proof of Theorem~\ref{thm:tifs-ns-lift}]
Each user will obtain one share from its own cache and one from each of its
\(h\) neighboring senders.  We therefore encode each file so that any
\(h+1\) shares recover it.
Use the alphabet matching in Remark~\ref{rem:tifs-finite-alphabet} to run
\(m\) independent copies of the given kernel.  If an original local file has
entropy \(B_0\), each extended local file has entropy \(B=mB_0\).  Form an
original file \(W_n\) from \(h+1\) such blocks and use them as the coefficients
of a polynomial of degree at most \(h\) in each field component.  Evaluating
at \(K\) distinct field points gives MDS shares
\(Y_{n,1},\ldots,Y_{n,K}\), each of entropy \(B=F/(h+1)\).
Any \(h+1\) shares recover \(W_n\).  Each sender's shares are independent across files,
so they form a valid local library for the repeated kernel.  Sender \(k\)'s
local library is
\(\mathcal Y_k=\{Y_{n,k}:n\in[N]\}\).  In the local instance centered at
sender \(k\), apply the \(K_0=h\) receiver shared-link kernel to receiver set
\(\cN_r(k)\setminus\{k\}\), using the projected demand \(\bd^{(k)}\).
Let \(C_u^{(k)}\) be the local cache symbol
assigned to receiver \(u\).

User \(u\) stores its home shares \(\{Y_{n,u}:n\in[N]\}\) and all
\(C_u^{(k)}\) for \(k\in\cN_r(u)\setminus\{u\}\).  Its cache entropy obeys
\[
H(Z_u)\le\frac{NF}{h+1}+h\frac{M_0F}{h+1}
=\frac{N+hM_0}{h+1}F.
\]
For demand \(\bd\), sender \(k\) evaluates the local delivery map from
\((\mathcal Y_k,\bd^{(k)})\) and broadcasts \(X_k^\bd\), with
\(H(X_k^\bd)\le R_0F/(h+1)\).  Thus
\[
\sum_{k=1}^K H(X_k^\bd)\le\frac{KR_0}{h+1}F.
\]

User \(u\) recovers \(Y_{d_u,k}\) from every local instance centered at
\(k\in\cN_r(u)\setminus\{u\}\) and has \(Y_{d_u,u}\) in its cache.  These
\(h+1\) MDS shares reconstruct \(W_{d_u}\).
\end{IEEEproof}
The map in \eqref{eq:tifs-ns-lift} scales a change in local memory by
\(h/(h+1)\) and a change in local rate by \(K/(h+1)\).  A componentwise
improvement in \((M_0,R_0)\) therefore improves the corresponding D2D point.

\begin{remark}[Packet count inherited from the local kernel]
Within each repeated kernel copy, the packets of an MDS share retain the
selected shared-link kernel's indices.  For a fixed outer alphabet
implementation, a kernel with fewer packets per local file, including a
PDA-based design, uses fewer packets per MDS share.
\end{remark}
\subsection{MAN Specialization and Performance Guarantees}

The MAN shared-link kernel gives an explicit family of D2D operating
points.  For \(h\) receivers and \(\tau=0,\ldots,h\), its memory--rate pair is
\[
M_0=\frac{N\tau}{h},\qquad
R_0=\frac{h-\tau}{\tau+1}.
\]
Each file is split into \(\binom h\tau\) packets indexed by \(\tau\)-subsets
of receivers.  Receiver \(j\in[h]\) caches the packets whose index contains \(j\).
Delivery sends one finite-field sum for every \((\tau+1)\)-subset.
For \(t=1,\ldots,h\), set \(\tau=t-1\) and denote the memory and load
obtained from Theorem~\ref{thm:tifs-ns-lift} by
\(M_t^{\rm bb}\) and \(R_t^{\rm bb}\).  The MAN-lift points are
\begin{equation}
    \left(M_t^{\rm bb},R_t^{\rm bb}\right)
    =
    \left(\frac{tN}{h+1},\frac{K(h+1-t)}{(h+1)t}\right),
    \qquad t=1,\ldots,h,
\label{eq:tifs-bb-points}
\end{equation}
together with the full-cache point \((N,0)\).  Let \(R_{\rm bb}(M)\) be the
lower convex envelope of these points.  Appendix~\ref{app:tifs-ns-example}
gives the complete placement, delivery, and decoding for a six-user instance.
\begin{theorem}[Block-boundary converse]
\label{thm:tifs-ns-converse}
Under the single-shot, non-interactive, noiseless orthogonal local-broadcast
model of Section~\ref{sec:system-model}, assume \(N\ge K\ge h+1\).  Every
reliability-only cyclic D2D scheme satisfies
\begin{equation}
R^\star(M,r)\ge
\max_{1\le s\le K-1}
\frac{K}{\delta_s}
\left[s\left(1-\frac{sM}{N}\right)\right]^+,
\label{eq:tifs-ns-converse}
\end{equation}
where \(\delta_s=\min\{h,K-s\}\).
\end{theorem}

The \(\delta_s\) external senders carry the information missing from a block
of \(s\) user caches.  Appendix~\ref{app:tifs-ns-converse} proves the bound
using file-index symmetry and cyclic incidence counting.
\begin{corollary}[Exact high-memory interval]
\label{cor:tifs-ns-high}
If \(K\ge h+1\), then for
\(M\in[hN/(h+1),N]\),
\[
R^\star(M,r)=\frac{K}{h}\left(1-\frac{M}{N}\right).
\]
\end{corollary}

\begin{theorem}[Factor-\(9/2\) guarantee]
\label{thm:tifs-ns-gap}
Under the same single-shot non-interactive sum-load model, assume
\(N\ge K\ge h+1\).  For every \(M\in[N/(h+1),N]\),
\[
R_{\rm bb}(M)\le \frac92 R^\star(M,r).
\]
\end{theorem}

Appendix~\ref{app:tifs-ns-gap} proves the factor-\(9/2\) guarantee by comparing
the MAN-lift memory-sharing load with the block-boundary converse.
\begin{IEEEproof}[Proof of Corollary~\ref{cor:tifs-ns-high}]
Memory sharing between the last nontrivial MAN-lift point
\[
\left(\frac{hN}{h+1},\frac{K}{h(h+1)}\right)
\]
and the full-cache point \((N,0)\) gives the upper bound.  For the lower bound,
use the \(N\) all-equal demand vectors
\(\bd^{(n)}=(n,\ldots,n)\), \(n\in[N]\).  For any user \(i\), reliability over
these demands and the fact that its own transmission is a function of its
cache give
\[
H\!\left(W_{[N]}\given
Z_i,\{X_{\cG_i^{(1)}}^{\bd^{(n)}}:n\in[N]\}\right)=0.
\]
Since \(|\cG_i^{(1)}|=h\) for \(K\ge h+1\),
\[
(N-M)F\le
\sum_{n=1}^N\sum_{k\in\cG_i^{(1)}}H(X_k^{\bd^{(n)}}).
\]
Every transmitter belongs to exactly \(h\) singleton boundary sets.
Summing the preceding inequality over all users therefore yields
\[
K(N-M)F\le h\sum_{n=1}^N\sum_{k=1}^K
H(X_k^{\bd^{(n)}})\le hNRF.
\]
Therefore
\[
R^\star(M,r)\ge \frac{K}{h}\left(1-\frac{M}{N}\right),
\]
matching the memory-sharing load throughout the stated interval.
\end{IEEEproof}
\FloatBarrier

\section{Conclusion}

Our construction turns secure shared-link caching schemes into
radius-limited D2D schemes. The cache budget includes each sender's library
representation, retained random state, and delivery keys. The secrecy proof
handles the dependence between senders' shares and shows that each user's
cache and audible broadcasts jointly reveal no information about
nonrequested files. The converse counts broadcasts crossing a user block
and uses one user's secrecy constraint to remove the broadcasts already
included in its observation.

With \(h=2r\) and \(K\ge h+1\), the PCC lift attains the minimum load
\(K/h\) for every \(M\ge Nh\) when \(N>1\). For \(N\ge K\), its
memory-sharing envelope is within a factor \(15\) of the optimum on
\([N/h+2,Nh]\).
The same factor-\(15\) guarantee holds for every reception radius.

\appendices
\section{Proof of Theorem~\ref{thm:tifs-secure-converse}}
\label{app:tifs-secure-converse}
By Lemma~\ref{lem:tifs-symmetrization}, take a cyclic-user and file-index
symmetric scheme with the same memory--load pair.  Fix
\(s\in[K-1]\).  For each cyclic shift \(i\in[K]\), define the contiguous
block of \(s\) users by
\[
    \cB_i^{(s)}:=\{i,i+1,\ldots,i+s-1\}
\]
and let \(\cG_i^{(s)}\) be the set of outside transmitters heard by at least
one user in that block:
\[
\cG_i^{(s)}:=
\{k\notin\cB_i^{(s)}:\cN_r(k)\cap\cB_i^{(s)}\ne\emptyset\}.
\]
On the cycle,
\[
    |\cG_i^{(s)}|=\delta_s=\min\{h,K-s\}.
\]
Let \(q=q_s=\lfloor N/s\rfloor\).  We choose demands so that each block
recovers \(sq\) distinct files and the transmission entropies can be aligned
across cyclic shifts.  Start with \(q\) base demand vectors
\(\bd^{(1)},\ldots,\bd^{(q)}\) as follows.  In each vector all \(K\) users
request distinct files, and across the \(q\) vectors the users in
\(\cB_1^{(s)}\) request \(sq\) distinct files.  Assign these \(sq\le N\)
block requests first.  In each demand vector, assign the remaining \(K-s\)
users distinct files from the \(N-s\ge K-s\) indices unused by the block.

Let \(g_i\) be the cyclic user shift that maps \(\cB_1^{(s)}\) to
\(\cB_i^{(s)}\), and use
\(\widetilde{\bd}_i^{(j)}:=g_i\bd^{(j)}\) for block \(i\).  The users in
\(\cB_i^{(s)}\) then request \(sq\) distinct files across the \(q\) instances.
Let \(\mathcal W_i^{(s,q)}\) collect these \(sq\) random files.  Their
indices are fixed by the chosen demands, and entropy of the collection
means joint entropy.  Under both \(\widetilde{\bd}_i^{(j)}\) and
\(\bd^{(j)}\), all \(K\) users request distinct files.  Hence there is a
permutation \(\pi_{i,j}\) of the file indices satisfying
\(\pi_{i,j}(\widetilde{\bd}_i^{(j)})=\bd^{(j)}\).  File-index symmetry gives,
for every transmitter \(k\),
\begin{equation}
H\!\left(X_k^{\widetilde{\bd}_i^{(j)}}\right)
=H\!\left(X_k^{\bd^{(j)}}\right).
\label{eq:tifs-file-align-secure}
\end{equation}
With the block index \(i\) fixed, abbreviate
\(X_{\mathcal A}^{(j)}
:=X_{\mathcal A}^{\widetilde{\bd}_i^{(j)}}\).

For a block user \(u\in\cB_i^{(s)}\), let \(\cH_i^{(s)}(u)\) be the set
of boundary transmitters heard by \(u\):
\[
\cH_i^{(s)}(u):=\cG_i^{(s)}\cap \cN_r(u).
\]
The largest possible value of
\(|\cH_i^{(s)}(u)|\) over \(u\in\cB_i^{(s)}\) is
\[
\eta_s=\min\{K-s,\max\{h/2,h+1-s\}\}.
\]
To see this, number the block positions by \(j=0,\ldots,s-1\).  The number
of outside transmitters heard at position \(j\) is
\[
\min\{K-s,[r-j]^++[r-(s-1-j)]^+\}.
\]
The sum of the two positive parts is convex in \(j\), so its maximum is
attained at an endpoint, where it equals
\(r+[r+1-s]^+=\max\{h/2,h+1-s\}\).
Choose an anchor \(u_1^\star\in\cB_1^{(s)}\) that hears \(\eta_s\)
outside transmitters and set
\(u_i^\star=g_i u_1^\star\).  Fix one demand instance index \(q_0\in[q]\).
Let \(W^\star\) be the file requested by
\(u_i^\star\) in instance \(q_0\).  Let \(\overline{\mathcal W}\) collect
the remaining \(sq-1\) random files:
\[
\overline{\mathcal W}:=\mathcal W_i^{(s,q)}\setminus\{W^\star\}.
\]
Then \(H(\overline{\mathcal W}\given W^\star)=(sq-1)F\) because the files are
independent.

For each of the \(q\) demand instances, all users in \(\cB_i^{(s)}\) decode
their requested files from their caches, transmissions from users inside
\(\cB_i^{(s)}\), and transmissions from \(\cG_i^{(s)}\).  Since delivery is
non-interactive, all transmissions by users inside \(\cB_i^{(s)}\) are
functions of their caches and the demand.  Therefore the files in
\(\mathcal W_i^{(s,q)}\), and in particular \(\overline{\mathcal W}\), are
recoverable from
\[
Z_{\cB_i^{(s)}}\quad\text{and}\quad
\{X_{\cG_i^{(s)}}^{(j)}:j=1,\ldots,q\},
\]
where \(j\) indexes the demand instance.  Hence
\[
H\!\left(\overline{\mathcal W}\given
Z_{\cB_i^{(s)}},X_{\cG_i^{(s)}}^{(1)},\ldots,
X_{\cG_i^{(s)}}^{(q)},W^\star\right)=0.
\]
It follows that
\[
\begin{aligned}
(sq-1)F
&=I\!\left(\overline{\mathcal W};
Z_{\cB_i^{(s)}},X_{\cG_i^{(s)}}^{(1)},\ldots,X_{\cG_i^{(s)}}^{(q)}
\given W^\star\right).
\end{aligned}
\]

File secrecy lets us remove the anchor's observation when bounding the
information needed to recover \(\overline{\mathcal W}\).  To separate this
observation from the remaining caches and transmissions, define
\(A_i^\star\) as the anchor's cache and the boundary broadcasts it hears in
instance \(q_0\):
\[
A_i^\star:=\bigl(Z_{u_i^\star},X_{\cH_i^{(s)}(u_i^\star)}^{(q_0)}\bigr),
\]
and let \(B_i^\star\) collect the remaining observations: the other \(s-1\)
caches, all boundary transmissions in the \(q-1\) instances other than
\(q_0\), and, in instance \(q_0\), the distinct transmissions indexed by
\(\cG_i^{(s)}\setminus\cH_i^{(s)}(u_i^\star)\).
Each retained boundary transmission appears once in \(B_i^\star\).
Together, \(A_i^\star\) and \(B_i^\star\) contain every block cache and
boundary transmission used by the decoder.  Consequently,
\[
H\!\left(\overline{\mathcal W}\given A_i^\star,B_i^\star,W^\star\right)=0,
\]
and hence
\[
(sq-1)F
=
I\!\left(\overline{\mathcal W};A_i^\star,B_i^\star\given W^\star\right).
\]
Reliability makes \eqref{eq:tifs-file-sec} equivalent to its conditioned
form \eqref{eq:tifs-file-sec-conditioned}.  All files in
\(\overline{\mathcal W}\) are unrequested by \(u_i^\star\) in instance
\(q_0\).  Thus the anchor's complete local view is independent of
\(\overline{\mathcal W}\) conditioned on \(W^\star\).  Since
\(A_i^\star\) is part of that view, data processing gives
\[
I\!\left(\overline{\mathcal W};
A_i^\star
\given W^\star\right)=0.
\]
By the chain rule,
\[
\begin{aligned}
I\!\left(\overline{\mathcal W};A_i^\star,B_i^\star\given W^\star\right)
&=
I\!\left(\overline{\mathcal W};A_i^\star\given W^\star\right) \\
&\quad +
I\!\left(\overline{\mathcal W};B_i^\star\given A_i^\star,W^\star\right) \\
&=
I\!\left(\overline{\mathcal W};B_i^\star\given A_i^\star,W^\star\right) \\
&\le H(B_i^\star).
\end{aligned}
\]
Subadditivity and \(H(Z_u)\le MF\) give
\[
\begin{aligned}
(sq-1)F
&\le (s-1)MF
 +\sum_{\substack{j\in[q]\\j\ne q_0}}
    \sum_{k\in\cG_i^{(s)}}
    H\!\left(X_k^{\widetilde{\bd}_i^{(j)}}\right)\\
&\quad+
   \sum_{k\in\cG_i^{(s)}\setminus
   \cH_i^{(s)}(u_i^\star)}
   H\!\left(X_k^{\widetilde{\bd}_i^{(q_0)}}\right).
\end{aligned}
\]

Sum over \(i\) and apply \eqref{eq:tifs-file-align-secure}.  For each fixed
\(j\ne q_0\), every transmitter appears in exactly \(\delta_s\) cyclic
shifts of the boundary set, giving
\[
\sum_{i=1}^K\sum_{k\in\cG_i^{(s)}}
H\!\left(X_k^{\widetilde{\bd}_i^{(j)}}\right)
=\delta_s\sum_{k=1}^K H(X_k^{\bd^{(j)}}).
\]
In the anchor instance \(q_0\), the sets
\(\cH_i^{(s)}(u_i^\star)\) are cyclic shifts of one set of cardinality
\(\eta_s\).  Thus each transmitter is omitted in exactly \(\eta_s\) shifts,
and
\[
\sum_{i=1}^K\sum_{k\in
\cG_i^{(s)}\setminus\cH_i^{(s)}(u_i^\star)}
H\!\left(X_k^{\widetilde{\bd}_i^{(q_0)}}\right)
=(\delta_s-\eta_s)\sum_{k=1}^K H(X_k^{\bd^{(q_0)}}).
\]
For each fixed base demand, the sum of all transmission entropies is at most \(RF\).  Hence the total
entropy of the retained boundary transmissions over all shifts is at most
\((q\delta_s-\eta_s)RF\), and therefore
\[
K(sq-1)F\le
K(s-1)MF+
(q\delta_s-\eta_s)RF.
\]
Dividing by \(K\) gives
\[
(sq-1)F\le
(s-1)MF+
(q\delta_s-\eta_s)\frac{RF}{K}.
\]
If \(q\delta_s>\eta_s\), rearranging and using \(R\ge0\) gives the
positive-part bound \eqref{eq:tifs-secure-converse}.  If
\(q\delta_s=\eta_s\), the transmission term vanishes, giving
\((sq-1)\le(s-1)M\).

For comparison with the PCC load, the next lemma gives a simpler form of
the secure converse for small user blocks.
\begin{lemma}[Small-block secure lower bound]
\label{lem:supp-small-block-lb}
Assume \(N\ge K\ge h+1\).  For any integer \(1\le s<h/2\), put
\(q=\lfloor N/s\rfloor\) and define the cache-dependent quantity
\(\Gamma_s\) by
\[
\Gamma_s:=\frac{(s-1)M+1}{N-s}.
\]
Whenever \(\Gamma_s<1\),
\begin{equation}
R_{\rm sec}^\star(M,r)
\ge
\frac{Ks}{h}\left(1-\frac{(s-1)M+1}{N-s}\right).
\label{eq:supp-small-block-lb}
\end{equation}
\end{lemma}

\begin{IEEEproof}
Because \(1\le s<h/2\) and \(K\ge h+1\),
\[
\eta_s=h+1-s,
\qquad
\delta_s=\min\{h,K-s\}\ge h+1-s=\eta_s.
\]
Since \(q\ge2\), this block size satisfies \(q\delta_s>\eta_s\).  The
converse denominator obeys
\[
q\delta_s-\eta_s\le qh-\eta_s<qh.
\]
Let \(A=(s-1)M+1\).  Since
\(sq>N-s\) and \(\Gamma_s=A/(N-s)<1\), we have \(A<sq\), so the converse
numerator \(sq-A\) is positive.  Theorem~\ref{thm:tifs-secure-converse} then
gives
\begin{align*}
R_{\rm sec}^\star(M,r)
&\ge \frac{K(sq-A)}{qh}
 =\frac{Ks}{h}\left(1-\frac{A}{sq}\right)\\
&\ge \frac{Ks}{h}\left(1-\frac{A}{N-s}\right),
\end{align*}
where the last inequality uses \(sq>N-s\) and \(A\ge0\).
\end{IEEEproof}

\section{Proof of Theorem~\ref{thm:tifs-pcc-gap}}
\label{app:tifs-pcc-gap}
We compare the PCC load with the secure converse on each interval between
adjacent PCC memory points.  Under \(N\ge K\ge h+1\), these coordinates are strictly
increasing.  For \(u=1,\ldots,h-2\),
\[
M_{u+1}^{{\rm PCC}\to r}-M_u^{{\rm PCC}\to r}
=
\frac{N(h+1)}{(h-u)(h-u+1)}
-\frac{1}{u(u+1)}
>0,
\]
and the final difference satisfies
\[
M_h^{{\rm PCC}\to r}-M_{h-1}^{{\rm PCC}\to r}
=\frac{N(h+1)}2-1-\frac1{h-1}>0.
\]
At \(M=M_h^{{\rm PCC}\to r}\), Corollary~\ref{cor:tifs-pcc-endpoint}
gives \(R_{{\rm PCC}\to r}^{\rm up}(M)=R_{\rm sec}^\star(M,r)=K/h\).
Every
\(M\in[M_1^{{\rm PCC}\to r},M_h^{{\rm PCC}\to r})\) lies in a segment
\[
M\in[M_u^{{\rm PCC}\to r},M_{u+1}^{{\rm PCC}\to r}]
\]
for some \(u\in\{1,\ldots,h-1\}\).  Memory sharing gives
\begin{equation}
R_{{\rm PCC}\to r}^{\rm up}(M)\le \frac Ku.
\label{eq:tifs-pcc-up-proof}
\end{equation}

If \(u\ge h/15\), then the \(s=1\) secure converse gives
\(R_{\rm sec}^\star(M,r)\ge K/h\), and therefore
\[
\frac{R_{{\rm PCC}\to r}^{\rm up}(M)}
{R_{\rm sec}^\star(M,r)}
\le \frac hu\le 15.
\]
For the remaining integers \(1\le u<h/15\), the evenness of \(h\) implies
\(h\ge16\).  Put \(v=u+1\) and \(H=h-u=h-v+1\).  The segment gives
\[
M\le M_{u+1}^{{\rm PCC}\to r}=M_v^{{\rm PCC}\to r}
= \frac{Nv}{H}+1+\frac1v.
\]
Write \(M_v=M_v^{{\rm PCC}\to r}\).  For an integer \(1\le s<h/2\), define
\[
\Gamma:=\frac{(s-1)M+1}{N-s}.
\]
Whenever \(\Gamma<1\), Lemma~\ref{lem:supp-small-block-lb} yields
\begin{equation}
R_{\rm sec}^\star(M,r)
\ge \frac{Ks}{h}(1-\Gamma).
\label{eq:pcc-gap-basic-15}
\end{equation}

\emph{Case \(3\le u<h/15\).}  Choose
\[
x:=\frac{H}{3v},\qquad s:=\lfloor x\rfloor.
\]
Then \(x>3\), so
\[
s\ge x-1=\frac{h-4u-3}{3(u+1)},
\qquad s\le x<h/2.
\]
Since \(s-1\le x-1\), \(N-s\ge N-x>0\), and \(M\le M_v\),
\[
\Gamma
\le
\frac{(x-1)M_v+1}{N-x}.
\]
Expanding the numerator gives
\begin{align*}
\frac{(x-1)M_v+1}{N}
&=
\left(\frac13-\frac{u+1}{h-u}\right)
\\
&\quad+
\frac{(h-4u-3)(u+2)}{3(u+1)^2N}
\,+\,\frac1N \\
&\le \frac13+\frac5{48}
=\frac7{16}.
\end{align*}
The last inequality uses \(H<N\), \(v\ge4\), and the identity
\[
\begin{aligned}
-\frac{v}{H}
+\frac{(H-3v)(v+1)}{3v^2N}+\frac1N
&=\frac{H}{N}\frac{v+1}{3v^2}-\frac{v}{H}-\frac1{vN}\\
&\le\frac{v+1}{3v^2}\le\frac5{48}.
\end{aligned}
\]
Furthermore,
\[
\frac{x}{N}\le \frac{1}{3(u+1)}\le \frac1{12},
\]
so
\[
\Gamma\le \frac{7/16}{1-1/12}=\frac{21}{44},
\qquad
1-\Gamma\ge\frac{23}{44}.
\]
Substituting this estimate into \eqref{eq:pcc-gap-basic-15} gives
\[
R_{\rm sec}^\star(M,r)
\ge
\frac{K}{h}\cdot
\frac{h-4u-3}{3(u+1)}\cdot
\frac{23}{44}.
\]
Combining this lower bound with \eqref{eq:tifs-pcc-up-proof} gives
\[
\frac{R_{{\rm PCC}\to r}^{\rm up}(M)}{R_{\rm sec}^\star(M,r)}
\le
\frac{132h(u+1)}{23u(h-4u-3)}.
\]
For fixed \(u\), the right side decreases with \(h\).  The constraint
\(u<h/15\) gives \(h\ge15u+1\), so the ratio is at most
\[
\frac{132(15u+1)(u+1)}{23u(11u-2)}<15.
\]
The strict inequality is equivalent to
\[
1815u^2-2802u-132>0,
\]
which holds for every \(u\ge3\): the polynomial is positive at \(u=3\) and
strictly increases thereafter.

\emph{Case \(u=2\).}  The condition \(h>30\) and the evenness of \(h\) give
\(h\ge32\).  Choose
\[
x:=\frac{h-2}{9},\qquad s:=\lfloor x\rfloor.
\]
Then \(s\ge x-1=(h-11)/9\), \(s<h/2\), and
\[
M\le M_3^{{\rm PCC}\to r}=\frac{3N}{h-2}+\frac43.
\]
The inequalities \(s-1\le x-1\) and \(N-s\ge N-x>0\) give
\[
\Gamma\le \frac{(x-1)M_3+1}{N-x}.
\]
For fixed \(h\), define the positive coefficients \(\alpha\) and \(\beta\) by
\[
\alpha=\frac{h-11}{3(h-2)}>0,
\qquad \beta=\frac{4h-17}{27}>0.
\]
The upper bound on \(\Gamma\) is then \((\alpha N+\beta)/(N-x)\).
This rational function has derivative
\(-(\beta+\alpha x)/(N-x)^2<0\) with respect to \(N\).
It is therefore largest under \(N\ge h+1\) at \(N=h+1\), giving
\[
\Gamma\le
\frac{13h^2-115h-65}{3(h-2)(8h+11)},
\]
and hence
\[
1-\Gamma\ge
\frac{11h^2+100h-1}{3(h-2)(8h+11)}>0.
\]
Applying \eqref{eq:pcc-gap-basic-15} yields
\[
R_{\rm sec}^\star(M,r)
\ge
\frac{K}{h}\cdot\frac{h-11}{9}\cdot
\frac{11h^2+100h-1}{3(h-2)(8h+11)}.
\]
Consequently,
\[
\frac{R_{{\rm PCC}\to r}^{\rm up}(M)}{R_{\rm sec}^\star(M,r)}
\le
\frac{27h(h-2)(8h+11)}{2(h-11)(11h^2+100h-1)}.
\]
For every \(h\ge32\), this expression is below \(15\).  Clearing its positive
denominator gives
\[
38h^3-165h^2-10812h+110>0,
\]
whose left side is positive at \(h=32\) and has positive derivative throughout
\(h\ge32\).

\emph{Case \(u=1\).}  The first segment gives
\[
M\le M_2^{{\rm PCC}\to r}=\frac{2N}{h-1}+\frac32,
\qquad
R_{{\rm PCC}\to r}^{\rm up}(M)\le K.
\]
It remains to show \(R_{\rm sec}^\star(M,r)\ge K/15\) on this segment.
Choose \(s=\lfloor h/6\rfloor\).  Since \(h\ge16\), we have \(s\ge2\) and
\(s<h/2\).  Let \(q=\lfloor N/s\rfloor\).  The relation
\(N\le s(q+1)-1\) gives
\[
M\le \frac{2(s(q+1)-1)}{h-1}+\frac32.
\]
It suffices to prove
\begin{equation}
\begin{aligned}
&15\left[sq-1-(s-1)\left(\frac{2(s(q+1)-1)}{h-1}
+\frac32\right)\right] \\
&\hspace{6em}\ge qh-(h+1-s).
\end{aligned}
\label{eq:u1-floor-check}
\end{equation}
Let \(D(q)\) denote the left side of \eqref{eq:u1-floor-check} minus its
right side.  The function \(D(q)\) is affine in \(q\).  Denote its
coefficient of \(q\) by \(C\):
\[
C=15s-\frac{30s(s-1)}{h-1}-h.
\]
Write \(h=6a+b\), with \(a=\lfloor h/6\rfloor=s\) and
\(b\in\{0,2,4\}\).  Since \(h\ge16\), the admissible ranges are
\(a\ge3\) for \(b=0,2\) and \(a\ge2\) for \(b=4\).  Then
\[
C(h-1)=
\begin{cases}
24a^2+21a, & b=0,\\
24a^2+27a-2, & b=2,\\
24a^2+33a-12, & b=4,
\end{cases}.
\]
Each polynomial is positive on its admissible range, so \(C>0\).
Since \(N\ge h+1\) gives \(q\ge6\),
\(D(q)\ge D(6)\).  Direct expansion gives
\[
2(h-1)D(6)=
\begin{cases}
18a^2+509a-77, & b=0,\\
18a^2+535a-63, & b=2,\\
18a^2+561a-129, & b=4,
\end{cases}.
\]
Each polynomial is positive on its admissible range, proving
\eqref{eq:u1-floor-check}.  Since \(qh-(h+1-s)>0\), this inequality and
the upper bound on \(M\) give
\[
15\bigl[sq-1-(s-1)M\bigr]\ge qh-(h+1-s)>0.
\]
Here \(\eta_s=h+1-s\) and \(\eta_s\le\delta_s\le h\), so the secure
converse yields
\[
\begin{aligned}
R_{\rm sec}^\star(M,r)
&\ge\frac{K\bigl[sq-1-(s-1)M\bigr]}{q\delta_s-\eta_s}\\
&\ge\frac{K\bigl[sq-1-(s-1)M\bigr]}{qh-(h+1-s)}
\ge\frac K{15}.
\end{aligned}
\]
Together with \(R_{{\rm PCC}\to r}^{\rm up}(M)\le K\), this completes the
\(u=1\) case and the proof.

\section{Worked Non-Secure MAN-Lift Example}
\label{app:tifs-ns-example}

\begin{example}[A radius-limited non-secure MAN lift]
Take \(K=6\), \(r=2\), \(h=4\), and \(N=3\).  For each file \(W_n\), form
six MDS shares \(Y_{n,1},\ldots,Y_{n,6}\), each of size \(F/5\), such that
any five shares recover the file.  Sender \(k\)'s local receiver set is
\(\mathcal K_k=\cN_r(k)\setminus\{k\}\).  Use a four-receiver MAN kernel
with parameter \(\tau=1\).  Split every local share \(Y_{n,k}\) into four
packets \(Y_{n,k}^{u}\), indexed by \(u\in\mathcal K_k\).
Receiver \(u\) caches the packet indexed by \(u\) for every file and
neighboring sender.  Each user
also caches its \(N\) complete home shares \(\{Y_{n,u}:n\in[N]\}\).

For demand \(\bd=(1,2,3,1,2,3)\), sender \(k\) transmits, for every pair
\(\{u,v\}\subseteq\mathcal K_k\), the finite-field sum
\[
X_{k;uv}=Y_{d_u,k}^{v}+Y_{d_v,k}^{u}.
\]
User \(1\), for example, hears senders \(2,3,5,6\).  In sender \(2\)'s
instance, the sums for pairs \(\{6,1\}\), \(\{1,3\}\), and
\(\{1,4\}\) reveal \(Y_{1,2}^{6}\), \(Y_{1,2}^{3}\), and
\(Y_{1,2}^{4}\), because user \(1\) has the companion packets indexed by
\(1\).
The three decoded packets and the cached packet \(Y_{1,2}^{1}\) recover
the local MDS share \(Y_{1,2}\).  Repeating the same step for senders
\(3,5,6\) gives user \(1\) all four non-home MDS shares.  Its home share
\(Y_{1,1}\) supplies the fifth share, so user \(1\) reconstructs \(W_1\).
Every user applies the same local
subtraction and MDS decoding steps.

Each packet has size \(F/20\).  A user stores \(3F/5\) bits of home shares
and \(4\cdot3F/20=3F/5\) bits from its four neighboring instances, for a
total of \(6F/5\) bits.  Each sender transmits \(\binom{4}{2}=6\) packets, and
the six senders transmit \(36F/20=9F/5\) bits in total.  These packet counts
give the achievable peak-load budget below, with local parameters
\(M_0=N/4\) and \(R_0=3/2\):
\[
(M,R)=\left(\frac65,\frac95\right).
\]
\end{example}

\section{Proof of Theorem~\ref{thm:tifs-ns-converse}}
\label{app:tifs-ns-converse}

Apply Lemma~\ref{lem:tifs-symmetrization} and fix a scheme with cyclic-user
and file-index symmetry at the same memory--load pair.
Fix \(s\in[K-1]\).  For each cyclic shift \(i\), let
\(\cB_i^{(s)}\) be the contiguous block of \(s\) users and
\(\cG_i^{(s)}\) the set of outside transmitters heard by at least one user in
the block, as in Appendix~\ref{app:tifs-secure-converse}.  Then
\(|\cG_i^{(s)}|=\delta_s=\min\{h,K-s\}\).

Because \(N\ge K\), choose a fixed demand vector \(\bd\) under which all
\(K\) users request distinct files.  For every block \(\cB_i^{(s)}\),
let \(\cD_i^{(s)}=\{d_u:u\in\cB_i^{(s)}\}\subseteq[N]\) be its set of
\(s\) requested file indices.  The corresponding random file tuple is
\(W_{\cD_i^{(s)}}\).
Non-interactive delivery makes the transmissions sent by users inside
\(\cB_i^{(s)}\) functions of the block caches and the demand.  Correctness
therefore gives
\[
H(W_{\cD_i^{(s)}}\given Z_{\cB_i^{(s)}},X_{\cG_i^{(s)}}^{\bd})=0,
\]
and hence
\begin{align*}
H(W_{\cD_i^{(s)}}\given Z_{\cB_i^{(s)}})
&\le H(X_{\cG_i^{(s)}}^{\bd}\given Z_{\cB_i^{(s)}})\\
&\le\sum_{k\in\cG_i^{(s)}}H(X_k^{\bd}).
\end{align*}

We next derive a lower bound on the uncertainty in the requested files
using the uncertainty remaining in the whole library after observing the
block caches.
For this fixed block, write \(Z=Z_{\cB_i^{(s)}}\).  Let
\(\mathcal A\subseteq[N]\) be a set of \(s\) file indices, and let
\([0]=\emptyset\), so \(W_{[0]}\) is the empty tuple.  The chain rule
and conditioning give
\[
\begin{aligned}
H(W_{\mathcal A}\given Z)
&=\sum_{j\in\mathcal A}
 H(W_j\given Z,W_{\mathcal A\cap[j-1]})\\
&\ge\sum_{j\in\mathcal A}
 H(W_j\given Z,W_{[j-1]}).
\end{aligned}
\]
In the sum over all such \(\mathcal A\), each
index \(j\) occurs \(\binom{N-1}{s-1}\) times.  Dividing by \(\binom Ns\)
and applying the chain rule to \(W_{[N]}\) yields
\[
\frac{1}{\binom Ns}\sum_{|\mathcal A|=s}
H(W_{\mathcal A}\given Z)
\ge
\frac{s}{N}H(W_{[N]}\given Z).
\]
Since the block contains \(s\) caches,
\(H(Z_{\cB_i^{(s)}})\le sMF\), and thus
\[
H(W_{[N]}\given Z_{\cB_i^{(s)}})\ge NF-sMF.
\]
File-index symmetry makes \(H(W_{\mathcal A}\given Z)\) equal for all
\(s\)-element index sets \(\mathcal A\).  The average lower bound therefore
applies to \(\mathcal A=\cD_i^{(s)}\), yielding
\[
H(W_{\cD_i^{(s)}}\given Z_{\cB_i^{(s)}})
\ge s\left(1-\frac{sM}{N}\right)F.
\]
Conditional entropy is nonnegative, so the right-hand side can be replaced
by its positive part.

Across all cyclic shifts \(i=1,\ldots,K\), each transmitter appears in
exactly \(\delta_s\) boundary sets.  The same fixed demand \(\bd\) serves every
block, so incidence counting gives
\[
\delta_s\sum_{k=1}^K H(X_k^{\bd})
\ge
Ks\left(1-\frac{sM}{N}\right)^+F.
\]
The peak-load constraint \(\sum_k H(X_k^{\bd})\le RF\) then yields
\eqref{eq:tifs-ns-converse}.

\section{Proof of Theorem~\ref{thm:tifs-ns-gap}}
\label{app:tifs-ns-gap}

We compare the MAN-lift memory-sharing load with the converse throughout
the MAN-lift memory interval.  Let \(x=(h+1)M/N\), so that
\(x\in[1,h+1]\).  For each integer \(t=1,\ldots,h+1\), let \(b_t\)
be the load at the MAN-lift point \(x=t\), divided by \(K\).  Thus
\(b_t=1/t-1/(h+1)\), and \(t=h+1\) is the full-cache point.
The slopes between successive points are
\[
b_{t+1}-b_t=-\frac1{t(t+1)},
\]
which increase strictly with \(t\).  Thus every MAN-lift point lies on the
lower convex hull.  For \(x\in[t,t+1]\), \(t=1,\ldots,h\), let
\(u_t(x)\) denote the normalized load obtained by interpolation between
adjacent points.  Then
\[
\frac{R_{\rm bb}(M)}{K}
=
u_t(x):=
\frac{2t+1-x}{t(t+1)}-\frac1{h+1}.
\]
The indices \(1\le s\le h\) satisfy \(s\le K-1\), and
\(\delta_s\le h\).  Define the normalized lower bound \(v_h(x)\) by
restricting Theorem~\ref{thm:tifs-ns-converse} to these indices and replacing
\(\delta_s\) by \(h\).  This gives
\[
\frac{R^\star(M,r)}{K}
\ge
v_h(x):=
\max_{1\le s\le h}
\frac{s}{h}\left(1-\frac{sx}{h+1}\right)^+.
\]
To obtain an explicit lower bound on \(v_h(x)\), fix
\(x\in[1,h+1)\) and consider the quadratic function \(f\) of a real
variable \(z\), defined by
\[
f(z):=z\left(1-\frac{zx}{h+1}\right).
\]
Its maximizer is \(y=(h+1)/(2x)\), with maximum
\(f(y)=(h+1)/(4x)\).  Completing the square gives
\[
f(z)=f(y)-\frac{x}{h+1}(z-y)^2.
\]
Let \(s(x)\) be the nearest integer to \(y\), with ties rounded up:
\[
s(x)=\left\lfloor y+\frac12\right\rfloor
=\left\lfloor\frac{h+1}{2x}+\frac12\right\rfloor .
\]
Then \(1\le s(x)\le h\) and \(|s(x)-y|\le1/2\), so
\[
f(s(x))\ge \frac{h+1}{4x}-\frac{x}{4(h+1)}
=\frac{(h+1)^2-x^2}{4x(h+1)}.
\]
Therefore
\[
v_h(x)\ge
\frac{(h+1)^2-x^2}{4xh(h+1)}.
\]

\emph{Region I: \(x\in[1,2]\).}
Here \(t=1\).  Direct algebra reduces
\[
u_1(x)<\frac92\frac{(h+1)^2-x^2}{4xh(h+1)}
\]
to
\[
9((h+1)^2-x^2)-4xh((h+1)(3-x)-2)>0.
\]
Put \(A=h+1\ge3\) and denote the left side by \(Q_A(x)\).  Its quadratic
coefficient is \(4A(A-1)-9>0\), and its global minimum is
\[
\min_{x\in\mathbb R}Q_A(x)
=\frac{84A^3-229A^2+80A-16}{4A^2-4A-9}>0.
\]
To verify positivity, set \(z=A-3\ge0\).  The numerator becomes
\[
84z^3+527z^2+974z+431>0.
\]
The denominator is positive for \(A\ge3\), proving the bound in Region I.

\emph{Region II: \(2\le t\le (h+1)/4\) and \(x\in[t,t+1]\).}
The lower bound on \(v_h(x)\) gives
\[
\frac{u_t(x)}{v_h(x)}
\le
\frac{4x(2t+1-x)}{t(t+1)}
\frac{h(h+1)}{(h+1)^2-x^2}.
\]
The first factor is a concave quadratic in \(x\), maximized at
\(x=t+1/2\), and hence is at most
\[
\frac{4(t+1/2)^2}{t(t+1)}
=4+\frac{1}{t(t+1)}
\le \frac{25}{6},
\]
because \(t\ge2\).  For the second factor,
\(x\le t+1\le (h+1)/4+1\).  The conditions \(t\ge2\) and \(h=2r\)
give \(h\ge8\).  Direct algebra gives
\[
\frac{h(h+1)}{(h+1)^2-((h+1)/4+1)^2}<\frac{27}{25}.
\]
After clearing the positive denominator, the inequality becomes
\(5(h+1)^2+184(h+1)-432>0\), which holds throughout this region.  Hence the
product is below \((25/6)(27/25)=9/2\).

\emph{Region III: \(t\ge (h+1)/4\).}
Let \(w_1(x)\) denote the lower bound obtained by choosing \(s=1\)
in the definition of \(v_h(x)\).  Then
\[
v_h(x)\ge w_1(x):=\frac{1-x/(h+1)}{h}.
\]
For \(x\in[t,t+1]\) with \(x<h+1\),
\[
\frac{d}{dx}\frac{u_t(x)}{w_1(x)}
=-\frac{h(h+1-t)(h-t)}{t(t+1)(h+1-x)^2}\le0.
\]
The ratio is therefore largest at \(x=t\), giving
\[
\frac{u_t(x)}{v_h(x)}
\le
\frac{u_t(x)}{w_1(x)}
\le
\frac{(1/t)-1/(h+1)}{(1-t/(h+1))/h}
=\frac{h}{t}
\le \frac{4h}{h+1}<4.
\]
The three regions establish
\[
R_{\rm bb}(M)\le \frac92 R^\star(M,r)
\]
for \(M<N\).  At \(M=N\), caching the entire library gives
\(R_{\rm bb}(N)=R^\star(N,r)=0\), which proves the endpoint directly.

\end{document}